\documentclass[%
 reprint,
 amsmath,amssymb,
 aps,prd,
]{revtex4-2}

\usepackage{graphicx}
\usepackage{dcolumn}
\usepackage{bm}
\usepackage{hyperref}
\usepackage{amsthm}
\usepackage{braket}

\theoremstyle{plain}
\newtheorem{proposition}{Proposition}
\newtheorem{theorem}[proposition]{Theorem}
\newtheorem{lemma}[proposition]{Lemma}
\newtheorem{corollary}[proposition]{Corollary}

\theoremstyle{definition}
\newtheorem{definition}[proposition]{Definition}
\theoremstyle{remark}
\newtheorem*{remark}{Remark}

\begin{document}

\preprint{APS/123-QED}

\title{Convergence of combinatorial gravity}

\author{Christy Kelly}
\email{ckk1@hw.ac.uk}
\author{Fabio Biancalana}%
 \email{F.Biancalana@hw.ac.uk}
\affiliation{School of Engineering and Physical Sciences, Heriot-Watt University}

\author{Carlo Trugenberger}
 \email{ca.trugenberger@bluewin.ch}
\affiliation{SwissScientific, Geneva}%

\date{\today}

\begin{abstract}
We present a new regularisation of Euclidean Einstein gravity in terms of (sequences of) graphs. In particular, we define a discrete Einstein-Hilbert action that converges to its manifold counterpart on sufficiently dense random geometric graphs (more generally on any sequence of graphs that converges suitably to the manifold in the sense of Gromov-Hausdorff). Our construction relies crucially on the Ollivier curvature of optimal transport theory. Our methods also allow us to define an analogous discrete action for Klein-Gordon fields. These results are part of the ongoing programme combinatorial approach to quantum gravity where we seek to generate graphs that approximate manifolds as metric-measure structures. 
\end{abstract}

\maketitle


\section{\label{sec:Intro}Introduction}

While Gauss discovered intrinsic geometry with his \textit{theorema egregium} in 1827, and Riemann intuited manifolds in the 1850s, it was not until the 1930s that a more or less modern mathematical definition of a differentiable manifold was made \cite{Gauss_Curvature,Riemann_Hypotheses,VeblenWhitehead_FoundationsDiffGeo}; also c.f. \cite{Weyl_Riemann,Whitney_Imbedding} for important contributions to this line of development and, e.g. \cite{Berger_Panoramic,Spivak_Comprehensive} for more recent reviews. The intervening period saw the development of general relativity by Einstein (and others) \cite{Einstein_Meaning} which was a fantastic corroboration of Gauss' original intuition---exemplified by his famous measurement of the large Brocken-Hohehagen-Inselberg triangle---that the geometry of space was a matter of empirical determination.\footnote{Historical honesty demands that we note that this was probably never intended as a test of spacetime geometry, the Euclidean structure of which, it seems, Gauss never seriously doubted. He of course recognised the logical possibility of non-Euclidean geometries and remarked that the measurement could be regarded as a corroboration of the Euclidean nature of space, but it seems the measurement was required for more mundane reasons relating to the geodetic survey of Hanover Gauss had been commissioned to carry out. See \cite{Breitenberger_Gauss} for more details.} Indeed the general programme of relativity theory surely represents one of the highpoints of the intersection between physics and geometry in the modern period.
	
Recently there has been active mathematical research in an area which might be loosely---and somewhat paradoxically---called \textit{discrete differential geometry} examining discrete analogues of smooth notions, driven by applications in computer science---especially computer graphics and mesh processing \cite{BobenkoEtAl_DDG,CraneWardetzky_GlimpseDDG,Crane_IntroDDG}, but also more loosely as a natural extension of methods of discrete topology (discrete Morse theory and combinatorial algebraic topology \cite{Scoville_DiscreteMT,Kozlov_CombAT}) to e.g. topological data analysis \cite{Carlsson_TopologyData}---network geometry \cite{Bianconi_Challenges,BogunaEtAl_NetworkGeometry,NajmanRomon_DiscreteCurvature} and quantum gravity. Indeed, as a quite general principle, it is desirable to find coarse formalisms for gravity since quantum fluctuations of spacetime are expected to ruin smooth structure. In approaches where spacetime is fundamentally discrete, this coarseness obviously must be promoted to full discreteness \cite{Oriti_AppQG}; even without fundamental (physical) discreteness, however, it may nonetheless be desirable to formulate a discrete---and not simply coarse---approximation of gravity as a nonperturbative regularisation of the smooth theory in the context of the asymptotic safety scenario \cite{Eichhorn_AsymptoticSafety}. Indeed, in a gravitational context, the use of discrete methods dates back to at least Regge's classic paper \cite{Regge_GRwithoutCoords} where he introduced the eponymous calculus that allowed for the calculation of curvature in terms of deficit angles on manifold triangulations. Regge's somewhat heuristic account has been rigorously confirmed by Cheeger, M\"{u}ller and Schrader in their demonstration of the convergence of curvature on suitable triangulations in a manifold \cite{CheegerMullerSchrader_CurvatyrePFSpace}. This has led to the development of simplicial formalisms for quantum gravity \cite{Hamber_QG}, the foremost of which is the dynamical triangulations approach \cite{AGJL,Loll_CDTReview,ADJ}.
	
	The Euclidean dynamical triangulations approach \cite{ADJ}---often in the form of a matrix model \cite{DiFrancescoGinspargZinnJustin}---was assiduously pursued in the 1980s and 1990s in two-dimensions, after it became clear that the scaling limit of this model was quantum Liouville theory \cite{KPZ_Fractal2dQG, David_CFT, DistlerKawai_CFT2DQG}; following an observation of Polyakov, this made it simultaneosuly a regularisation of 2D-gravity coupled to conformal matter and of noncritical string theory \cite{Polyakov_QuantumGeometryBosonicStrings}. Taking a gravitational interpretation, it was realised in the process of this work that the emergent structure of $2D$-quantum Euclidean spacetime was that of a \textit{Brownian sphere}, a topological $2$-sphere with Hausdorff dimension $4$ and spectral dimension $2$ \cite{AmbjornWatabiki_ScalingQG,AmbjornEtAl_SD2dQG}. These findings have recently been placed on a firm mathematical footing: there is a large body of literature in this direction, but most relevant results are referred to in \cite{Miller-LQGReview} which focuses on the key proof of the equivalence between quantum Liouville gravity and the Brownian sphere; for rigorous spectral dimension results see \cite{Lee-ConformalGrowthRates,GwynneMiller-RandomWalkRandomPlanarMaps}. In other dimensions, however, the Euclidean dynamical triangulations approach alternates between a highly crumpled phase of infinite Hausdorff dimension and a \textit{branched polymer phase} of topological dimension $1$, Hausdorff dimension $2$ and spectral dimension $4/3$ \cite{ADJ,GurauRyan_MelonsBranchedPolymers,Aldous_CRTI,Aldous_CRTII,Aldous_CRTIII,ADJ_SummingGenera,JonssonWheater_SDBP}. Branched polymers are generally regarded as a pathological model of quantum spacetime and for this reason the Euclidean dynamical triangulations formalism has typically been seen as inadequate for a general treatment of quantum gravity. This perspective was only compounded by the realisation that the phase transition in Euclidean dynamical triangulations is first-order \cite{BialasEtAl_DTPhaseTransition,Bakker_DTPhaseTransition,RindlisbacherForcrand_FirstOrder}. The \textit{causal} dynamical triangulations \cite{AGJL,Loll_CDTReview} programme appears to do much better in this regard. In this approach one typically fixes not only the topology of spacetime but also a preferred foliation. These models typically have a rich phase structure and have much improved behaviour with regards to the structure of the scaling limit \cite{AJL_3dCDT,AmbjornEtAl_NewPhase,AmbjornEtAl_2DCDTHL,AmbjornEtAl_PhaseStructure,DurhuusJonssonWheater_SDCDT}. 
	
	An alternative approach to quantum gravity that also makes much of the causal structure of spacetime is \textit{causal set theory}; see \cite{Surya_CSTReview} for a recent review. The basic insight is that geometric structures on spacetime may be encoded in causal relations and related topologies, an insight that in modern form dates back to at least Zeeman in the 1960s, who showed that the Lorentz group (augmented by dilatations) was the group of causal automorphisms on Minkowski space and found a suitable topology to encode this information \cite{Zeeman_CausalityLorentzGroup,Zeeman_TopologyMinkowski}. A decade later, these results were vastly generalised by Hawking, King and McCarthy \cite{HawkingKingMcCarthy} and Malament \cite{Malament_CST} who showed that causal structure on a Lorentzian manifold encodes both the differential and conformal structure. This work motivates the causal set \textit{Hauptvermutung} which purports that any Lorentzian manifold should be characterised by an essentially unique causal set, i.e. a locally finite poset describing the causal structure of spacetime. A significant development in this regard is the demonstration that the \textit{Benincasa-Dowker action}---essentially defined by counting short causal chains in a causal set---converges to the Einstein-Hilbert action on causal sets generated by Poisson processes in Lorentzian manifolds (`sprinklings') \cite{BenincasaDowker_ScalCurvCausSet}.
	
	Returning to a Euclidean setting, a major recent development in Riemannian geometry using ideas from optimal transport theory \cite{Villani_OptimalTransport,Villani_Topics} has been a synthetic characterisation of Ricci curvature using the metric-measure structure of the manifold \cite{LottVillani,Sturm_GeomI,Sturm_GeomII,Ohta_MeasureContraction, Ollivier_RCMCMS,Ollivier_RCMS, SturmVonRenesse_TransportInequalities,Villani_OptimalTransport,NajmanRomon_DiscreteCurvature}: a Riemannian manifold $\mathcal{M}$ is a metric space when equipped with the geodesic distance $\rho$; it is made into a metric-measure space when equipped with a random walk, i.e. a probability measure at each point. This insight has allowed for the definition of coarse notions of curvature valid in generic metric-measure spaces. One notion due to Sturm \cite{Sturm_GeomI} and independently Lott and Villani \cite{LottVillani} is related to the so-called \textit{$L_2$-transport cost}; this is perhaps the canonical mathematical example of a synthetic curvature, but is ill-equipped for use in discrete spaces. On the other hand an alternative notion due to Ollivier \cite{Ollivier_RCMCMS,Ollivier_RCMS,Ollivier_Survey} is well-defined for discrete metric-measure spaces and indeed has widely gained traction for a variety of applications in the network theory community: c.f. e.g. \cite{FarooqEtAl_Brain,JostLiu_RicciCurv,Ni_RicIntTop,Sandhu_Cancer,Sandhu_Market,SiaEtAl_CommunityDetection,TannenbaumEtAl_Cancer,WangEtAl_DiffGeom,WangEtAl_Interference,WangEtAl_QUBO,WangEtAl_WirelessNetwork,WhiddenMatsen_SubtreeGraph}. The fundamental intuition captured by the Ollivier curvature is that in a positively curved space, the average distance between two nearby balls will be closer than their centres. This entire line of development is a natural extension of the programme of metric geometry \cite{Gromov_MetricStructures,BuragoBuragoIvanov_MetricGeometry}, a coarse generalisation of many ideas in Riemannian geometry using the fact that many results of the smooth theory rely only on its structure as a metric space or length space. We give a slightly more formal presentation of optimal transport theory and the Ollivier curvature in section \ref{section: MathematicalPreliminaries} below.
	
	In the context of Euclidean gravity the potential ramifications are clear: the Ollivier curvature may be used to specify a new regularisation of Euclidean Einstein gravity in terms of discrete structures such as graphs, regarded as discrete metric spaces; concretely, the aim is to specify a discrete action defined on graphs in terms of the Ollivier curvature which will approximate the manifold Einstein-Hilbert action whenever the graphs themselves are a good approximation for the manifold. (This latter constraint of course places limitations on the types of space that can be approximated: we shall only be concerned with compact---and \textit{ipso facto} finite diameter---spaces throughout.) We have begun to pursue this line of thought in \cite{Trugenberger_CombQG,KellyEtAl,KellyEtAl_Circle}, calling the associated Ollivier curvature based model combinatorial quantum gravity due to a formal analogy with Einstein gravity. (Other than our own work, \cite{KlitgaardLoll_HowRound,KlitgaardLoll_ImplementingQuantumRicciCurvature,KlitgaardLoll_QuantumRicciCurvature}, \cite{BrunekreefLoll-CurvatureProfile} and \cite{Gorard} are all examples of similar works using Ollivier curvature in the context of quantum gravity). The basic result of this paper---presented in section \ref{section: DiscreteAction}---is that, for a slightly different action to the one adopted in \cite{KellyEtAl,KellyEtAl_Circle,Trugenberger_CombQG}, the formal analogy can be made precise at the classical level in terms of a convergence result for a particular discrete action. This is of course necessary if we wish to pursue a putative dynamical quantum model. More concretely, in theorem \ref{theorem: MainTheorem}, we show that there is a discrete action defined on graphs that converges to the Euclidean Einstein-Hilbert action on any (compact) manifold for a sequence of graphs converging to the manifold in a sense to be discussed at length below.
	
	We should stress that this result does not give a full characterisation of combinatorial \textit{quantum} gravity insofar as it says very little about the overall partition function of quantum gravity. It merely shows that we can control the error in the phase associated to a manifold configuration, which follows trivially from the fact that we can control the error in the action. More precisely, if we imagine our partition function is a sum over (the limit points of) sequences of random graphs, those limit points which correspond to a manifold $\mathcal{M}$ will contribute $\exp(-\beta \mathcal{A}_{EH}(\mathcal{M}))$ to the Euclidean partition function, where $\mathcal{A}_{EH}$ is the Einstein-Hilbert action. This is surely a necessary condition for any proposed regularisation of gravity. We suggest it is also sufficient: away from the classical limit, it is precisely the wager of approaches of this kind that the `tail' dependence of the partition function on nonmanifold configurations might ameliorate some of the problems with the naive theory defined via a sum over smooth classical metrics. Indeed, the results of $2D$ Euclidean dynamical triangulations suggest that even the semiclassical limit displays some rather nonregular behaviour with regards to smooth structure. Still, it might be said that a regularisation of gravity needs also to respect minima (extrema) of the action, i.e. the `classical' (action minimising) configurations are dominated by `manifold-like' configurations. While we tend to agree with this assessment, we regard this as an essentially dynamical problem of the full quantum theory since it can be assessed by examination of the structure of configurations actually arising in some dynamical model. At any rate, apart from a particular case already considered \cite{KellyEtAl_Circle}, we are not currently capable of settling these questions analytically. Instead we would hope to extend some already promising, if preliminary, numerical results \cite{KellyEtAl,KellyEtAl_Circle} to the convergent model specified here. Note that these numerical results have been extended and clarified in \cite{GorskyValba_Thermofield} and \cite{Akara-pipattanaChotibutEvnin-Emergence}.
	
	There are some important caveats, however, to the simple application of Ollivier curvature as a direct graph-theoretic regularisation of the Ricci curvature: the metric-measure structure of a graph is certainly not equal to the metric-measure structure of a Riemannian manifold in general, and as such the Ollivier curvature evaluated in a graph will typically be quite different from the Ollivier curvature evaluated in a manifold. The upshot, of course, is that combinatorial quantum gravity as defined in \cite{Trugenberger_CombQG,KellyEtAl,KellyEtAl_Circle} does not converge to classical gravity in general; note that we believe these models may nonetheless have something relevant to say about Euclidean quantum gravity because we obtain useful results on the Ricci flat sector, where the action of combinatorial quantum gravity does agree with the classical action. Since the Einstein-Hilbert action is defined as the integral over spacetime of the scalar curvature $R$, it is clear that even if the difference between graph and manifold metric-measure structures can be overcome, specifying a convergent model of combinatorial gravity will involve both approximating spacetime integrals (via graph vertex sums) and finding some way of carrying out a trace at each point---since $R={\rm tr}({\rm Ric})$. 
	
	An important development in this direction was the recent demonstration that the Ollivier curvature of a random geometric graph in a Riemannian manifold converges to the Ollivier curvature of the underlying manifold \cite{Hoorn_Convergence,HoornEtAl_OllCurvConv}. The essential idea is that for a sufficiently dense sampling of a Riemannian manifold $\mathcal{M}$, a random geometric graph $G$ will approximate $\mathcal{M}$ as a metric space, while one can choose local measures that simultaneously converge to a random geometric graph in $G$. In this way the associated Ollivier curvature $\kappa_G^\delta(u,v)$ converges to the manifold Ollivier curvature $\kappa_{\mathcal{M}}^\delta$. At one level, the main result of this paper can be seen as a fairly elementary extension of this edge-curvature convergence result to show that we have a discrete Einstein-Hilbert action that converges on random geometric graphs.
	
	While convergence on random geometric graphs is essentially sufficient for a kinematic characterisation of a \textit{given} Riemannian manifold in terms of a graph, a key step in the proof of convergence relies heavily on the fact that random geometric graphs are obtained via Poisson processes in Riemannian manifolds, since it involves an extension of certain probabilistic grid matching results due to Talagrand and others; see \cite{Talagrand_UpperBounds} for a review of these results. In the context of (quantum) gravity, it is desirable for the limiting configuration to be determined \textit{dynamically}, and hence convergence of the Ollivier curvature on random geometric graphs in a \textit{fixed} Riemannian manifold is not sufficient. As such we will try to \textit{avoid} the assumption that our graphs are obtained from Poisson processes, instead preferring to think of them as arising from an as yet unspecified dynamical model of random graphs. Since this model is left unspecified, we are forced in this paper to deal with graphs that are not in fact random. Another aim of this paper is to show that there is a more general \textit{kinematic} context in which we may talk about both Ollivier curvature and Einstein-Hilbert action convergence for a given (compact) Riemannian manifold $\mathcal{M}$.
	
	This more general context is specified in terms of the \textit{Gromov-Hausdorff distance} \cite{Gromov_MetricStructures,BuragoBuragoIvanov_MetricGeometry}. We discuss this notion at more length in section \ref{section: MathematicalPreliminaries}, but it is worth making a few remarks here: the Gromov-Hausdorff distance is a metric $\rho_{GH}$ on the space of isometry classes of \textit{compact} metric spaces, and provides us with a notion of convergence for such spaces. In particular, any compact length space---metric space where the distance between two points is the shortest length of an admissible curve between those points---is obtained as the Gromov-Hausdorff limit of a graph. The basic idea in specifying the Gromov-Hausdorff distance is that for any metric space, there is a natural distance function (the \textit{Hausdorff distance}) on the compact subsets of the space in question; the Gromov Hausdorff distance between two compact metric spaces $X$ and $Y$ is then given by minimising the Hausdorff distance between $X$ and $Y$ over all isometric embeddings of $X$ and $Y$ into some arbitrary ambient space $Z$.
	
	As discussed at length in the appendix to \cite{KellyEtAl_Circle}, Gromov-Hausdorff distances and limits naturally respect gauge invariance because gauge transformations of the manifold structure (i.e. smooth local isometries with smooth inverses) turn out to be global length space isometries: any gauge transformation turns out to be a local \textit{metric} isometry---that is it preserves the distance function on certain open balls of a point---and any local metric isometry with continuous inverse is a global isometry. The Gromov-Hausdorff distance also realises a kind of background independence insofar as it is obtained by minimising the Hausdorff distance \textit{over all possible backgrounds}. At the same time, it will be helpful for the purposes of the results in this paper to treat the graphs as embedded in some background manifold. This can be done if we realise that $\rho_{GH}$ has another interpretation: the Gromov-Hausdorff distance $\rho_{GH}(X,Y)$ may be regarded as the obstruction to the existence of an isometric embedding $X\hookrightarrow Y$. In particular, in a precise sense to be specified below, the Gromov-Hausdorff distance $\rho_{GH}(X,Y)$ is small iff the spaces $X$ and $Y$ are \textit{nearly isometric}.
	
	Let us be very clear: strictly speaking our configurations are \textit{sequences} of (unlabelled weighted) graphs, but in fact the convergence result which holds for sequences follows from an approximation result for individual graphs. The only restriction we put on these graphs is that they are finite, simple, connected and Gromov-Hausdorff close to a given manifold. The discrete Einstein-Hilbert action on graphs may be defined in terms of purely combinatorial structures viz. the graph distance and graph Ollivier curvature, and so is specified for arbitrary graphs. The substance of our approximation result is then that \textit{if} a graph $G$ admits a nearly isometric embedding into a (compact) Riemannian manifold $\mathcal{M}$, i.e. if $G$ is Gromov-Hausdorff close to $\mathcal{M}$, then the discrete Einstein-Hilbert action on $G$ is close to its continuum counterpart. Note that it is crucial to show that this holds for \textit{arbitrary} nearly isometric embeddings in order to show that the approximation result is in some sense `generally covariant'; this is of course easily achieved if we leave the explicit embedding indeterminate throughout and only rely on the \textit{existence} of such an embedding.
	
	This result remains kinematic: we know very little about the kinds of configuration obtained from dynamical models based on the action introduced here. Such models may well result in pathologies. Evidently, our hope is that this is not in fact the case; we feel that our other results suggest \cite{Trugenberger_CombQG,Kelly_Exact,KellyEtAl_Circle} have given clear indications that there are is in fact something novel about Ollivier curvature driven random graph models, but clearly the question has yet to be settled. Nonetheless, we suggest that the broader kinematic interpretation proposed here---viz. considering sequences of graphs that converge to a given manifold in the sense of Gromov-Hausdorff---is preferable to a treatment in terms of random geometric graphs insofar as we do not see how the latter can be realised in terms of any dynamical model whereas the former is in principle compatible with all dynamical models. Even further, a dynamical model of random unweighted graphs based on a very similar Ollivier curvature based discrete Einstein-Hilbert action has been shown to generate configurations that converge in the sense of Gromov-Hausdorff to the circle \cite{KellyEtAl_Circle} after appropriate rescaling. If our proposed generalisation of the random geometric graph kinematics is valid, then we may in fact reinterpret the van der Hoorn et al. convergence result as a demonstration of the expected Gromov-Hausdorff proximity of random geometric graphs to their underlying manifolds for suitable parameter choices. 
	
	Finally let us say some words on the novelty and significance of our results: if we restrict ourselves to the context of random geometric graphs, the main result (theorem \ref{theorem: MainTheorem}) follows from some rather elementary considerations, given the convergence result of van der Hoorn et al \cite{Hoorn_Convergence,HoornEtAl_OllCurvConv}. Explicitly we show how to approximate traces and integrals on $\mathcal{M}$ via sums on the graph $G$ when $G$ is Gromov-Hausdorff proximal to $\mathcal{M}$; these results are fairly elementary and as such they are almost certainly not new, though we are not aware of their statement or application in the existing literature. The delicate nature of the proof is their combination: there are several different scales involved in the problem and each approximation result places different constraints on these scales. It is not at all clear that these constraints are consistent, and in fact the most naive approach to taking the trace is not consistent with the other constraints on the available scales. 
	
	On the other hand, if we are to consider the more general kinematic context we need to extend the van der Hoorn et al. convergence result to a geometric one. From a mathematical perspective the idea is to introduce a suitable topology on the space of (compact) metric-measure spaces in which the Ollivier curvature is `stable', i.e. curvature (bounds) respect limits in the relevant topologies. In fact, two relevant topologies already exist in the literature: in the first topology---in which Sturm-Lott-Villani curvature bounds are stable---one metric-measure space converges to another if it converges as a metric space in the sense of Gromov-Hausdorff and all the measures converge weakly to suitable corresponding measures as identified by (measurable) near isometries. In this topology the measures may converge arbitrarily slowly which unfortunately makes this topology unsuitable for present purposes: the requirement that the Ollivier curvature approximates the manifold Ricci curvature puts constraints on the rate at which the discrete Ollivier curvature is to converge if we wish to approximate the manifold Ricci curvature. Ollivier also introduced a topology on metric-measure spaces for which his curvature is stable: c.f. proposition 47 of \cite{Ollivier_RCMCMS}. Unfortunately the topology introduced here is so fine as to preclude even the van der Hoorn et al. derivation of convergence in random geometric graphs. Thus our `extension' of the van der Hoorn et al. result amounts to a demonstration that convergence in the Sturm-Lott-Villani topology can be made sufficiently rapid to ensure that the Ollivier curvature continues to approximate the Ricci curvature. This extension requires the application of known Euclidean semidiscrete optimal transport results \cite{HartmannSchumacher_Semidiscrete} and turns out to be more `costly' than the van der Hoorn et al. result in the sense that it imposes stricter constraints on the available scales of the problem. 
	
	Given that we expect the graph Ollivier curvature to converge to its manifold counterpart in suitable circumstances, out main result---that there is a convergent discrete Einstein-Hilbert action defined on networks---is then also no surprise. Indeed, had an analogous result not existed the Ollivier curvature would have been, at some level, the `wrong' notion. We have, of course, believed that the Ollivier curvature represents an interesting tool in this context for some time; nonetheless, as suggested in \cite{Akara-pipattanaChotibutEvnin-Emergence,KellyEtAl_Circle}, the precise relation between our previous work and quantum gravity proper was not entirely clear. We hope that this paper helps to clarify this issue.
\section{Mathematical Preliminaries}\label{section: MathematicalPreliminaries}
	The main purpose of this section of this section is to introduce the Ollivier curvature and the Gromov-Hausdorff distance since these ideas will play a basic role in the subsequent. However before beginning these more in depth discussions we shall make some points on notation and briefly review some simple ideas in metric and Riemannian geometry.
	
	\subsection{Some Points from Metric and Riemannian Geometry}
	We shall throughout be concerned with metric spaces. For a set $X$, a metric on $X$ will be denoted $\rho_X$ unless we specify otherwise. The \textit{open ball} of radius $r>0$ centred at a point $x\in X$ is denoted
	\begin{align}
	    B_r^X(x)=\set{y\in X:\rho_X(x,y)<r}.
	\end{align}
	The boundary of this ball is then given
	\begin{align}
	    \partial B_r^X(x)=\set{y\in X:\rho_X(x,y)=r}.
	\end{align}
	Later on we shall be concerned with the \textit{shell} or \textit{annulus} of \textit{radius} $R$ and \textit{width} $r$ centred at $x\in X$, which is defined as the set
	\begin{align}
	    S_{R,r}^X(x)=B_{R+r}^X(x)\backslash B_{R-r}^X(x).
	\end{align}
	
	We shall use a special notation for Euclidean space $\mathbb{R}^D$: the metric is denoted $\rho_D$, an open ball of radius $R$ and centre $x\in \mathbb{R}^D$ by $B_r^D(x)$, and the annulus of radius $R$, width $r$ and centre $x\in X$ by $S_{R,r}^D(x)$. We also let $B_r^D:=B_r^D(0)$ and $S_{R,r}^D:=S_{R,r}^D(0)$. The volume form on $\mathbb{R}^D$ will be denoted ${\rm vol}_D$. The volume  of the unit ball in $\mathbb{R}^D$ is denoted by $\omega_D$. The unit $(D-1)$-sphere is the subset of unit vectors in $\mathbb{R}^D$ and will be denoted $\mathbb{S}^{D-1}$. For any two points $x_1,\:x_2\in \mathbb{R}^D$, we let $\theta_{D}(x_1,x_2)$ denote the angle $\angle x_1Ox_2$ with $O$ the origin; by the cosine rule we have
	\begin{align}
	    \theta_D(x_1,x_2)=\arccos \left(\frac{||x_1||^2+||x_2||^2-||x_1-x_2||^2}{2||x_1||\cdot ||x_2||}\right).
	\end{align}
	
	It will be convenient to introduce the function $\mathfrak{d}~:~\mathbb{R}^D\backslash\set{0}~\rightarrow ~\mathbb{R}$ which gives the scale factor of the spherical volume element of the normalisation of its argument in Cartesian coordinates: let $x~=~(R,\varphi_1,...,\varphi_{D-1})\in \mathbb{R}^D\backslash \set{0}$ in spherical coordinates; then
	\begin{align}\label{equation: SphericalVolumeElement}
	        \mathfrak{d}(x)=\sin^{D-2}(\varphi_1)\sin^{D-3}(\varphi_2)\cdots \sin(\varphi_{D-2}).
	\end{align}
	We extend this function to all of $\mathbb{R}^D$ by choosing $\mathfrak{d}(0)=0$.
	
	$G$ will always denote a graph, by which we mean a finite, simple, connected, weighted graph. Also $N=|G|$ for the rest of the text. The weights will be given by a weight function $w_G:E(G)\rightarrow \mathbb{R}$, where $E(G)$ is the set of edges of $G$. $G$ may be regarded as a metric space, equipped with its geodesic distance: the length of a path $\gamma_{u,v}=v_0\cdots v_n$ between the vertices $u=v_0$ and $v=v_n$ is defined as
	\begin{align}
	    L(\gamma_{u,v})=\sum_{k=0}^{n-1}|w(v_kv_{k+1})|.
	\end{align}
	Then
	\begin{align}
	    \rho_G(u,v)=\inf_{\gamma_{u,v}}L(\gamma_{u,v})
	\end{align}
	where the infimum is taken over all paths $\gamma_{u,v}$ between $u$ and $v$. With this metric $G$ has the discrete topology, i.e. every subset of $G$ is open and hence Borel measurable. We shall assume $G$ is equipped with a natural measure of the volume of a subset given by the counting measure: $|E|$ is the number of points in $E\subseteq G$.
	
	Similarly, $\mathcal{M}$ will always denote a $D$-dimensional Riemannian manifold which is implicitly equipped with some Riemannian metric. This means in particular for each $p\in \mathcal{M}$ we have an inner product $\braket{\cdot,\cdot}_p$ on the tangent space $T_p\mathcal{M}$; we use the notation $||\cdot||_p$ for the associated norm. The \textit{angle} between two vectors $V_1,:V_2\in T_p\mathcal{M}$ is defined
	\begin{align}
	    \theta_p^{\mathcal{M}}(V_1,V_2)=\arccos\left(\frac{\braket{V_1,V_2}_p}{||V_1||_p\cdot ||V_2||_p}\right).
	\end{align}
	We assume that these local inner products vary smoothly over $\mathcal{M}$ and as such extend immediately to an inner product $\braket{\cdot,\cdot}$ on smooth vector fields of $\mathcal{M}$. Associated to the inner product is a unique metric compatible affine connection $\nabla$ known as the \textit{Levi-Civita connection}. This defines a notion of differentiation on vector fields on $\mathcal{M}$. Associated to this connection is also a natural volume form ${\rm vol}_{\mathcal{M}}$, and Riemmann, Ricci and scalar curvatures tensors, denoted ${\rm Riem}$, ${\rm Ric}$ and $R$ respectively. A subscript $p\in \mathcal{M}$ on any of these tensors refers to the restriction of the tensor to the point $p\in \mathcal{M}$. The (Riemannian) \textit{Einstein-Hilbert action} on $\mathcal{M}$ is defined
	\begin{align}\label{equation: EHAction}
	   \mathcal{A}_{EH}(\mathcal{M})&=\int_{\mathcal{M}}{\rm d vol}_{\mathcal{M}}(x)R(x).
    \end{align}
    This action, of course, governs gravitational dynamics in general relativity and is certainly well-defined and finite if $\mathcal{M}$ is compact.
	
	A smooth curve in $\mathcal{M}$ is a smooth mapping $\gamma:(a,b)\rightarrow \mathcal{M}$, where we assume $a<0<b$; we let the tangent vector to the curve at $t\in (a,b)$ be denoted by $\dot{\gamma}_t$. We say that a smooth curve $\gamma$ with domain $(a,b)$ \textit{connects} $p\in \mathcal{M}$ to $q\in \mathcal{M}$ iff we have unique $c,:d\in (a,b)$ such that $p=\gamma(c)$, $q=\gamma(d)$ and $c\leq d$. The \textit{length} of a curve $\gamma$ between $p$ and $q$ is then given
	\begin{align}
	    L_{p,q}(\gamma)=\int_{t=c}^{t=d}{\rm d t}\braket{\dot{\gamma}_t,\dot{\gamma}_t}_{\gamma(t)}.
	\end{align}
	$\gamma$ is a \textit{geodesic} iff $\nabla_{\dot{\gamma}_t}\dot{\gamma}_t=0$ for all $t\in{\rm dom}(\gamma)$; a geodesic is \textit{unit-speed} off $||\dot{\gamma}_t||_{\gamma(t)}=1$ for all $t\in {\rm dom}(\gamma)$. $\mathcal{M}$ is a metric space when equipped with the geodesic distance:
	\begin{align}
	    \rho_{\mathcal{M}}(p,q)=\inf L_{p,q}(\gamma)
	\end{align}
	where the infimum is taken over all curves that connect $p$ and $q$. It turns out that for any point $p\in \mathcal{M}$ and any $V\in T_p\mathcal{M}$, there is a maximal domain $(a,b)\subseteq \mathbb{R}$ such that there is a unique geodesic $\gamma$ satisfying ${\rm dom}(\gamma)=(a,b)$, $\gamma(0)=p$ and $\dot{\gamma}_0=V$. By making $||\dot{\gamma}_0||_p=V$ smaller we may extend the domain $(a,b)$ such that for $V$ sufficiently small $1\in (a,b)$. Since the tangent vector $V$ identifies the geodesic uniquely, we can thus define a mapping $\exp_p$ known as the \textit{exponential map} at $p\in \mathcal{M}$ via $\exp_p(V)=\gamma(1)$ where $\gamma$ is the unique geodesic such that $\gamma(0)=p$ and $\dot{\gamma}_0=V$. The domain of $\exp_p$ is some subset of $T_p\mathcal{M}$ that contains the origin. A metric ball in $T_p\mathcal{M}$ ($\mathcal{M}$) centred at the origin ($p\in \mathcal{M}$) will be called \textit{geodesic} iff it is a subset of the domain (codomain) of the exponential map at $p$. Note that for sufficiently small balls about the origin in $\mathbb{R}^D$, the exponential map is smooth and a radial isometry: the latter in particular means that $\rho_D(0,x)=\rho_{\mathcal{M}}(p,\exp_p(x))$ for all $x$ sufficiently close to the origin.
	
	Henceforth, we assume that $\mathcal{M}$ is complete, connected and without boundary. From section \ref{subsection: GHDistance} onwards, $\mathcal{M}$ will also be compact unless specified otherwise.  
	
	We shall use the big O notation---a little loosely---throughout to specify errors: for functions $f,\:g:U\rightarrow \mathbb{R}$, where $U\subseteq \mathbb{R}$ is some suitable subset, we write
	\begin{align}
	    f(x)=\mathcal{O}(g(x))
	\end{align}
	iff $|f(x)|<K g(x)$ for some constant $K>0$ in some suitable limit of $x$. We are always concerned with limits of $x$ such that $g(x)\rightarrow 0$: typically when we specify $g$ in terms of the variable $x=N\in \mathbb{N}$, we are interested in the limit $N\rightarrow\infty$ and we have a typical form $g(N)=N^{-a}$ for some $a>0$. Otherwise we shall typically specify $x$ in terms of variables such as $\delta,\:\varepsilon,\:\ell$ etc. and we are interested in limits as these variables go to $0$.
	\subsection{Ollivier Curvature}
	The Ollivier curvature utilises ideas from optimal transport theory \cite{Villani_OptimalTransport} for its construction so we begin with a review that follows closely the cited reference. We will also need some elementary ideas from measure theory and Riemannian geometry, e.g. c.f. \cite{Bogachev_MeasureI,DoCarmo_Riem} for nice reviews. Finally, we will also use ideas from \cite{Ollivier_RCMCMS} in this section without much comment.
	
	Let $(X,\rho_X)$ be a complete separable metric space and let $\mathcal{P}(X)$ denote the set of Borel probability measures on $X$; every measure considered in this text will be of this type, i.e. defined on the Borel $\sigma$-algebra of a complete separable metric space. Recall that for any measurable spaces $(\Omega_1,\Sigma_1)$ and $\Omega_2,\Sigma_2)$, the \textit{pushforward} of the measure $\mu:\Sigma_1\rightarrow \mathbb{R}$ under a measurable map $f:\Omega_1\rightarrow \Omega_2$ is defined as the measure
	\begin{align}
	    f_*\mu(E)=\mu(f^{-1}(E))
	\end{align}
	for all $E\in \Sigma_2$. The basic fact about pushforward measures is the following: let $f:\Omega_1\rightarrow \Omega_2$ be measurable and let $\mu$ be a Borel measure on $\Omega_1$. A mapping $g:\Omega_2\rightarrow \mathbb{C}$ is $f_*\mu$-integrable iff $g\circ f:X\rightarrow\mathbb{C}$ is $\mu$-integrable. Then 
	\begin{align}
	    \int_{\Omega_1}{\rm d}\mu(x)g(f(x))=\int_{\Omega_2}{\rm d} f_*\mu(y)g(y).
	\end{align}
	
	The basic notion in optimal transport theory is that of a \textit{transport plan}: for any $\mu,\:\nu\in \mathcal{P}(X)$, a \textit{transport plan} between $\mu$ and $\nu$ is a probability measure $\xi$ on $X\times X$ such that
	\begin{align}\label{equation: MarginalConstraint}
        (\pi_1)_*\xi=\mu && (\pi_2)_*\xi=\nu
    \end{align}
	where $\pi_1:(x,y)\mapsto x$ and $\pi_2:(x,y)\mapsto y$ are the projections onto the first and second elements respectively. We refer to the equations \ref{equation: MarginalConstraint} as \textit{marginal constraints} and $\mu$ and $\nu$ are said to be \textit{marginals} of the transport plan $\xi$. Roughly speaking, we may suppose we have a distribution $\mu$ of dirt on $X$ which we wish to transform into a distribution $\nu$; then given measurable subsets $E_1,\:E_2\subseteq X$, the idea is that $\xi(E_1\times E_2)$ denotes the amount of earth to be transported from $E_1$ to $E_2$ according to the transport plan $\xi$. The marginal constraints are required to ensure that we do indeed have the correct initial and final distributions and that no dirt is somehow lost in the process.
	
	The space of all transport plans between $\mu$ and $\nu$ is denoted $\Pi(\mu,\nu)$. The \textit{transport cost} of a transport plan $\xi\in \Pi(\mu,\nu)$ is defined 
	\begin{equation}
        \mathcal{T}_X(\xi)= \int_{X\times X}{\rm d} \xi(x,y)\rho_X(x,y),
	\end{equation}
	while the \textit{optimal transport cost} is then given
	\begin{equation}
        \mathcal{T}_X(\mu,\nu)= \inf_{\xi \in \Pi(\mu,\nu)}\mathcal{T}_X(\xi).
	\end{equation}
	This is also called the \textit{Wasserstein distance} between $\mu$ and $\nu$. A transport plan $\xi\in \Pi(\mu,\nu)$ is said to be \textit{optimal} iff $\mathcal{T}_X(\xi)=\mathcal{T}_X(\mu,\nu)$.
	
	It can be shown that optimal transport plans always exist. Moreover, the mapping $\mathcal{T}:(\mu,\nu)\mapsto \mathcal{T}(\mu,\nu)$ is a metric on $\mathcal{P}(X)$, i.e. $\mathcal{T}$ is positive definite, symmetric and subadditive. There is a slight caveat in that $\mathcal{T}(\mu,\nu)$ is not necessarily finite, but this problem can be avoided if we restrict consideration to spaces with bounded diameter or only consider measures with finite first moment. For the purposes of this paper, both restrictions can in fact be assumed to hold. With this in mind, we note that the Wasserstein distance metrises the topology of weak convergence on the space of all probability measures with finite first moment.
	
	A Riemannian $D$-manifold $\mathcal{M}$ can obviously be regarded as a metric space where the distance $\rho_{\mathcal{M}}(p,q)$ between two points $p,\:q\in \mathcal{M}$ is given by the length of the shortest geodesic between those points. Moreover, since $\mathcal{M}$ comes with a natural volume form ${\rm vol}$, we may define the uniform probability measures
	\begin{align}\label{equation: MeasureManifold}
	    \mu_p^\delta(E)=\frac{{\rm vol}(B^{\mathcal{M}}_\delta(p)\cap E)}{{\rm vol}(B^{\mathcal{M}}_\delta(p))},
	\end{align}
	for any pair $(p,\delta)\in \mathcal{M}\times (0,\infty)$, where $E$ is any (Borel) measurable set of $\mathcal{M}$. The idea is to specify the uniform measure with support given by (the closure of) the (open) ball $B^{\mathcal{M}}_\delta(p)$. Equipping $\mathcal{M}$ with such measures for each $p\in \mathcal{M}$ gives $\mathcal{M}$ the structure of a \textit{metric measure space}, i.e. a metric space equipped with a (Markovian) random walk.
	
	The basic connection between optimal transport and Ricci curvature arises in the following manner: let $p,\:q\in \mathcal{M}$ such that their geodesic distance $\ell= \rho_{\mathcal{M}}(p,q)$ is sufficiently small. For $\delta>0$ sufficiently small (i.e. for $\delta$ less than the injectivity radius), we may identify every point of $B_\delta^{\mathcal{M}}(p)$ with an element of $B_\delta^D\subseteq T_p\mathcal{M}\cong \mathbb{R}^D$ and similarly for $q$. Parallelly transporting $B_\delta^{\mathcal{M}}(p)$ along a minimal geodesic connecting $p$ and $q$ thus gives us a bijection $B_\delta^{\mathcal{M}}(p)\cong B_\delta^{\mathcal{M}}(q)$ which defines a so-called \textit{deterministic} transport plan $\xi_0\in \Pi(\mu_p^\delta,\mu_q^\delta)$. It turns out that this transport plan is in fact optimal---at least up to negligible errors---i.e. $\mathcal{T}_{\mathcal{M}}(\mu_p^\delta,\mu_q^\delta)=\mathcal{T}_X(\xi_0)$. A fairly basic application of the Jacobi equation, however, gives us the asympotic expression:
	\begin{align}
	\mathcal{T}_{\mathcal{M}}(\mu_p^\delta,\mu_q^\delta)=\ell\left(1+\frac{\delta^2{\rm Ric}_p(V,V)}{2(D+2)}+\mathcal{O}(\delta^2(\delta+\ell))\right),
	\end{align}
	where $V$ is the velocity of the unique unit-speed geodesic connecting $p$ and $q$ at $p$. Thus, if we define the manifold Ollivier curvature by
	\begin{align}\label{equation: ManifoldOllivierCurvature}
	    \kappa_{\mathcal{M}}^\delta(p,q)= 1-\frac{\mathcal{T}_{\mathcal{M}}(\mu_p^\delta,\mu_q^\delta)}{\rho_{\mathcal{M}}(p,q)},
	\end{align}
	we see that
	\begin{align}\label{equation: OllivierError}
	    \kappa_{\mathcal{M}}^\delta(p,q)=\frac{\delta^2 {\rm Ric}_p(V,V)}{2(D+2)}+\mathcal{O}(\delta^2(\delta+\ell)).
	\end{align}
	In this sense, up to a scale-dependent factor and small corrections, the manifold Ollivier curvature gives the local Ricci curvature.
	
	It should be clear that the manifold Ollivier curvature (Eq.~\ref{equation: ManifoldOllivierCurvature}) is determined by the metric-measure structure of $\mathcal{M}$ so we have the following generalisation: let $(X,\rho_X)$ be a complete separable metric space equipped with a random walk $\set{\mu_x}_{x\in X}$. The Ollivier curvature of $X$ is then defined as the function
	\begin{align}\label{equation: OllivierCurvature}
	    \kappa_X(x,y)=1-\frac{\mathcal{T}_X(\mu_x,\mu_y)}{\rho_X(x,y)},
	\end{align}
	on all sufficiently nearby but distinct points $x,\:y\in X$. Because the domain of $\kappa_X$ is not (necessarily) all of $X\times X$, it is convenient to regard $\kappa_X$ as a \textit{partial function} on $X\times X$.
	
	We finish this section by noting that the above definition applies fairly trivially to graphs. Each graph $G$ is regarded as a metric space as described above, so we are only concerned with the specification of the random walk on $G$. There is some freedom in the choice of graph measures; by analogy with the manifold measures Eq.~\ref{equation: MeasureManifold} we consider the uniform measures:
	\begin{align}\label{equation: GraphMeasure}
	    m^\delta_u(E)=\frac{|B^G_\delta(u)\cap E|}{|B^G_\delta(u)|},
	\end{align}
	for any $u\in G$ and all measurable $E\subseteq G$, where $\delta>0$ takes on any (small) strictly positive real-value. Note that since $G$ has the discrete topology, the Borel $\sigma$-algebra of $G$ is simply the set of all subsets of $G$, i.e. every subset of $G$ is measurable. In specifying the measures via Eq.~\ref{equation: GraphMeasure}, the hope is that the the uniform graph measures $m^\delta_u$ will approximate the uniform manifold measures $\mu^\delta_u$, allowing one to approximate the manifold Ollivier curvature by the graph Ollivier curvature. We show below that this is indeed the case. 
	\subsection{Gromov-Hausdorff Distance}\label{subsection: GHDistance}
	As mentioned in the introduction, the Gromov-Hausdorff distance between two compact metric spaces $X$ and $Y$ is defined by taking the infimum of the Hausdorff distance between the images of $X$ and $Y$ under isometric embeddings into some ambient metric space $Z$. Thus to introduce the Gromov-Hausdorff distance we begin by introducing the Hausdorff distance between subsets of a metric space. For more general references on the material in this section see \cite{BuragoBuragoIvanov_MetricGeometry,Gromov_MetricStructures}.
	
	Let $(X,\rho_X)$ be a metric space. For any point $x\in X$ and any set $A\subseteq X$, we have
	\begin{align}
	    \rho_X(x,A)= \inf_{y\in A}\rho_X(x,y).
	\end{align}
	The Hausdorff distance between two subsets $A,\:B\subseteq X$ is then defined
	\begin{align}
	    \rho_H^X(A,B)= \max\set{\sup_{x\in A}\rho_X(x,B),\sup_{y\in B}\rho_X(y,A)}.
	\end{align}
	
	There is an alternative formulation of the Hausdorff distance that permits us to introduce some useful notation. Let $B_\delta^X(x)$ be the (open) $\delta$-ball centred at $x$ for any $x\in X$. We define the \textit{$\varepsilon$-thickening} of any set $A\subseteq X$ as the set
	\begin{align}
	    A_\varepsilon= \bigcup_{x\in A}B_\varepsilon^X(x)
	\end{align}
	for any $\varepsilon>0$. With this notation we note that
	\begin{align}
	    \rho_H^X(A,B)=\inf\set{\varepsilon>0:B\subseteq A_\varepsilon{\rm  and }A\subseteq B_\varepsilon}.
	\end{align}
	
	The positivity, semidefiniteness and symmetry of $\rho_H$ are trivial consequences of the definition; note that symmetry requires both $B\subseteq A_\varepsilon$ and $A\subseteq B_\varepsilon$ since if $A\subseteq B$, $A\subseteq B_\varepsilon$ for all $\varepsilon>0$. For subadditivity we note that if $\rho_H^X(A,C)=\varepsilon$ and $\rho_H^X(C,B)=\delta$ we have $A\subseteq C_\varepsilon$ and $C\subseteq B_{\delta}$ and so $A\subseteq B_{\delta+\varepsilon}$ since $U\subseteq V$ implies $U_\varepsilon\subseteq V_\varepsilon$ for all $\varepsilon>0$ and all $U,\:V\subseteq X$ and $(U_\varepsilon)_{\delta}\subseteq U_{\delta+\varepsilon}$ for all $U\subseteq X$ and all $\delta,\:\varepsilon>0$. Similar remarks show that $B\subseteq A_{\delta+\varepsilon}$ and subadditivity follows from the infimum. 
	
	$\rho_H^X$ is thus a pseudometric on the set of subsets of $X$. To show that it is not a metric, let ${\rm cl}(A)$ denote the closure of $A$ and note that $B_\varepsilon(x)\cap A\neq \emptyset$ for any $\varepsilon>0$ for any $x\in {\rm cl}(A)\backslash A$. Hence ${\rm cl}(A)\subseteq A_\varepsilon$ for all $\varepsilon>0$, while $A\subseteq {\rm cl}(A)$, and $\rho_H^X(A,{\rm cl}(A))=0$. Taking $A$ not closed shows the failure of definiteness. $\rho_H^X$ thus defines a metric on the equivalence classes of subsets of $X$, where two subsets $A$ and $B$ of $X$ are equivalent iff $\rho_H^X(A,B)=0$. It turns out that we may represent these equivalence classes by the closed subsets of $X$: we have already shown that every set is equivalent to its closure so it is sufficient to show that two distinct closed sets are inequivalent. In particular suppose that $A,\:B\subseteq X$ are closed and $A\neq B$, i.e. without loss of generality we may assume that there is an $x\in A\backslash B$. Since $x\notin {\rm cl}(B)=B$, there is an $\varepsilon>0$ such that $B_\varepsilon^X(x)\cap A=\emptyset$; then $\rho_H^X(A,B)\geq \varepsilon>0$ and $A$ and $B$ are inequivalent as required.
	
	The Hausdorff distance is thus a metric on the closed subsets of a space; we show that it is a \textit{finite} metric on the closed and bounded subsets of $X$ if the metric on $X$ is finite: in particular suppose that $A$ and $B$ are bounded, i.e. $A\subseteq B_r(x)$ and $B\subseteq B_R(y)$ for some $r,\:R>0$ for all $x\in A$ and all $y\in B$. Since $\rho_X$ is finite we may pick $(x,y)\in A\times B$ so that $A\subseteq B_{r+R+\rho_X(x,y)}$ and $B\subseteq A_{r+R+\rho_X(x,y)}$ as required. Since every compact set is both closed and bounded, the Hausdorff dimension naturally induces a finite metric on the compact subsets of a space.
	
	With this in mind we automatically see that if we define the Gromov-Hausdorff distance $\rho_{GH}(X,Y)$ between two compact metric spaces $X$ and $Y$ as the infimum of the Hausdorff distance over isometric embeddings of $X$ and $Y$ into all possible ambient metric spaces $Z$, we have a finite pseudometric on the space of compact metric spaces, and which clearly gives $0$ iff $X$ and $Y$ are isometric.
	
	For the purposes of this paper we shall need a basic result on the characterisation of the Gromov-Hausdorff distance in terms of \textit{near isometries}. We need the following notions:
	\begin{enumerate}
	    \item Let $(X,\rho_X)$ be a metric space. An \textit{$\varepsilon$-net} in $X$ is a subset $A\subseteq X$ such that $A_\varepsilon=X$.
	    \item Now let $(X,\rho_X)$ and $(Y,\rho_Y)$ be metric spaces; we define the \textit{distortion} of a map $f:X\rightarrow Y$ as the quantity
	    \begin{align}
	        \text{dis}(f)=\sup_{(x_1,x_2)\in X\times X}|\rho_Y(f(x_1),f(x_2))-\rho_X(x_1,x_2)|.
	    \end{align}
	    \item An \textit{$\varepsilon$-isometry} between $X$ and $Y$ is a mapping $f:X\rightarrow Y$ such that ${\rm dis}(f)\leq \varepsilon$ and $f(X)$ is a $\varepsilon$-net in $Y$.
	\end{enumerate}
	The basic relation between near isometries and the Gromov-Hausdorff distance that we wish to show is that there is a $\mathcal{O}(\varepsilon)$-isometry between two compact metric spaces $(X,\rho_X)$ and $(Y,\rho_Y)$ iff $\rho_{GH}(X,Y)=\mathcal{O}(\varepsilon)$. More precisely
	\begin{lemma}\label{lemma: GHProximity=ExistenceNearIsometry}
	    Let $(X,\rho_X)$ and $(Y,\rho_Y)$ be compact metric spaces. If $\rho_{GH}(X,Y)\leq \varepsilon$ then there is a $2\varepsilon$-isometry $f:X\rightarrow Y$; conversely, given an $\varepsilon$-isometry $f:X\rightarrow Y$ we have $\rho_{GH}(X,Y)\leq 2\varepsilon$. 
	\end{lemma}
	This lemma is a standard result and the interested reader is directed to e.g. the monograph \cite{BuragoBuragoIvanov_MetricGeometry} for a proof; it allows us to control the precision of a nearly isometric embedding with the Gromov-Hausdorff distance between two metric spaces.
	
	\section{The Discrete Action}\label{section: DiscreteAction}
	In this section we prove our main result, theorem \ref{theorem: MainTheorem}. We first present the theorem and proof in a heuristic outline, before giving a more formal presentation. This section substantially overlaps with chapter 5 of the PhD thesis of one of the authors \cite{Kelly-Thesis} which gives more complete proofs of all the statements given here.
	\subsection{An Informal Outline of the Main Result}\label{subsection: InformalOutline}
	Our result concerns the following situation: recall that $G$ is a graph on $N$ vertices and $\mathcal{M}$ is a compact Riemannian manifold. We assume that we have an $\varepsilon$-isometry $\iota:G\hookrightarrow \mathcal{M}$. The assumption is that $G$ arises from some random dynamical process that does not fix the background manifold $\mathcal{M}$ \textit{a priori}; the graph $G$ itself however is a fixed configuration i.e. it has no random structure that we can appeal to. 
	
	At its most elementary our main claim is the following:
	\begin{quote}
	    Outline of Theorem \ref{theorem: MainTheorem}: \textit{If $\varepsilon=\varepsilon(N)$ is sufficiently small, $N$ sufficiently large and $\iota(G)$ samples $\mathcal{M}$ sufficiently evenly, then there is a discrete Einstein-Hilbert action $\mathcal{A}_{DEH}=\mathcal{A}_{DEH}(G)$ such that the error}
	    \begin{align}
	        \delta \mathcal{A}(G,\mathcal{M}):=|\mathcal{A}_{DEH}(G)-\mathcal{A}_{EH}(\mathcal{M})|
	    \end{align}
	    \textit{is small where $\mathcal{A}_{EH}$ is the Riemannian Einstein-Hilbert action (Eq.~\ref{equation: EHAction}). $\mathcal{A}_{DEH}$ is defined in terms of the graph Ollivier curvature and other graph-intrinsic quantities.}
	\end{quote}
	As suggested in the introduction, we have long believed that there should be some discrete Einstein-Hilbert action $\mathcal{A}_{DEH}$ such that the above is true and the main challenge is to find the form of $\mathcal{A}_{DEH}$ and specify any conditions that need to be satisfied for the above to hold. A convergence result also follows almost immediately from the above. 
	
	It is quite clear that going from the Einstein-Hilbert action to the graph Ollivier curvature requires three main steps:
	\begin{enumerate}
	    \item The Einstein-Hilbert action $\mathcal{A}_{EH}$ is defined as the integral over $\mathcal{M}$ of the scalar curvature $R$. Thus, if we can approximate integrals over $\mathcal{M}$ via a suitable operation in $G$, we reduce the problem of approximating $\mathcal{A}_{EH}$ to approximating the Ricci scalar $R$ at each point.
	    \item The Ricci scalar $R$ at a point $p\in \mathcal{M}$ is defined as the trace of the Ricci curvature ${\rm Ric}_p$ at $p\in \mathcal{M}$. If we can approximate the trace, the problem thus reduces to a finding an approximation of the Ricci curvature at each point $p\in \mathcal{M}$.
	    \item We know that the manifold Ollivier curvature approximates the Ricci curvature. It is thus sufficient for the purpose of theorem \ref{theorem: MainTheorem} to show that the graph Ollivier curvature approximates in some sense the manifold Ollivier curvature.
	\end{enumerate}
	It turns out that the difficulty associated with each of the three steps above is in ensuring that various graph theoretic measures approximate relevant measures on manifolds. This is no surprise in the case of the Ollivier curvature which is explicitly defined in terms of metric-measure structures. (Recall that in the basic set-up above we have \textit{assumed} that the graph $G$ approximates the manifold $\mathcal{M}$ as a metric space). This is also trivially the case in the first step where we wish to approximate an integral over $\mathcal{M}$; for the second step, however, the measure theoretic aspect of the problem comes from the fact that taking the trace of a bilinear form essentially involves averaging that bilinear form over the sphere.
	
	Why then is this a problem? The basic issue comes from the fact that a graph $G$ can metrically approximate a manifold $\mathcal{M}$ \textit{without} necessarily approximating $\mathcal{M}$ well at the measure theoretic level. For instance $\iota(G)$ can be regarded as a sample of points in $\mathcal{M}$ and this sample might be highly uneven with respect to the volume measure in $\mathcal{M}$. In this way natural operations such as sums over points in $G$ correspond to integrals in $\mathcal{M}$, \textit{distorted} by the sampling density. 
	
	At the most naive level these problems are circumvented in two ways. In the case of approximating the Ricci curvature and taking the trace, the manifold measures that we wish to approximate via graph structures are \textit{local} measures by which we mean that the measures are associated to a given point and supported in some small compact region of that associated point. This allows us to exploit the universal local structure of smooth manifolds and extract suitable local subgraphs of $G$ around each vertex in such a way that the local subgraphs sample the relevant regions of $\mathcal{M}$ evenly under the nearly isometric imbedding $\iota$. In this way, both the Ricci and scalar curvatures at a point $p\in \mathcal{M}$ can be obtained from the Ollivier curvature of a vertex $u\in G$, $p=\iota(u)$. On the other hand, to obtain the full Einstein-Hilbert action, one must integrate over $\mathcal{M}$. Without fixing $\mathcal{M}$ in advance, the only way to obtain a good approximation at this level is to impose a condition on the evenness of the sample $\iota(G)\subseteq \mathcal{M}$. While this can be done using entirely metric constraints---i.e. without reference to the measure theoretic structure of $G$ or $\mathcal{M}$---it imposes severe restrictions on the kinds of sequence of graph converging to $\mathcal{M}$ we can consider and imposes major restrictions on the scope of the present results.
	
	Let us consider each step in slightly more detail. We first consider step $1$, relating to integration over $\mathcal{M}$. A formal statement of the relevant result along with proof is given in appendix \ref{appendix: Integral}. Essentially the idea is to replace
	\begin{align}
	    \int_{\mathcal{M}}\text{d vol}_{\mathcal{M}}(x)f(x)\leftrightarrow \sum_{u\in G}\omega_D\varepsilon^Df(\iota(u)),
	\end{align}
	where the right-hand side is interpreted as a kind-of Riemann sum for the function $f$ on the cover $\set{B_{\varepsilon}^{\mathcal{M}}(\iota(u))}_{u\in G}$ of $\mathcal{M}$. In particular, the factor $\omega_D\varepsilon^D$ corresponds to the volume of the Euclidean $D$-ball of radius $\varepsilon$. Two restrictions have to be imposed however to make this approximation work well. Firstly, $f$ cannot vary too much over the balls $B_{\varepsilon}^{\mathcal{M}}(\iota(u))$, $u\in G$, if it is reasonable to set it constant over the entire ball. Concretely, this can be achieved by restricting the class of functions $f$ under consideration to e.g. those that restrict to $K$-Lipschitz functions on the balls for a suitable Lipschitz constant $K$. Since every smooth function is locally Lipschitz continuous, this restriction on the class of functions $f$ considered is sufficiently generous to allow the approximation to be valid for \textit{any} smooth function, as long as we are free to decrease $\varepsilon$ (increase $N$). The second condition relates to the fact that the cover $\set{B_{\varepsilon}^{\mathcal{M}}(\iota(u))}_{u\in G}$ is not pairwise disjoint so we need some way to control the overlap error. This is achieved, for instance, by assuming that the balls $\set{B_{\varepsilon(1-\alpha)}^{\mathcal{M}}(\iota(u))}_{u\in G}$ are in fact pairwise disjoint for some $\alpha\in (0,1)$. Note that we will find that we require $\alpha\ll 1$ which places a sever restriction on the class of nearly isometric imbeddings $\iota$ that can be considered. 
	
	Two basic steps are required in showing that we may approximate the trace. Fundamentally, we use the fact that the trace in $\mathbb{R}^D$ has the following integral representation:
	\begin{align}\label{equation: TraceIntegral}
	    {\rm tr}_{D}(T)&=\frac{1}{\omega_D}\int_{\mathbb{S}^{D-1}}{\rm d vol}_{\mathbb{S}^{D-1}}(x)T(x,x)\nonumber\\
	    &=\frac{1}{\omega_D}\int_{\mathbb{S}^{D-1}}{\rm d^D x}\:\mathfrak{d}(x) T(x,x).
	\end{align}
	(Recall that $\mathfrak{d}$ is a function that assigns to $x$ the Jacobean determinant for spherical coordinates evaluated on $x/||x||$; c.f. Eq.~\ref{equation: SphericalVolumeElement}.) This can be checked explicitly with an elementary, if tedious, calculation. The factor $\omega_D$ arises because the trace of the constant bilinear form $1:(x,x)\mapsto 1$ is $D$ and we have the relation $D\omega_D={\rm vol}_{D}(\mathbb{S}^{D-1}).$ In fact, with minimal adjustment we can clearly express Eq.~\ref{equation: TraceIntegral} as an integral over $\partial B_\ell^D$ for any $\ell>0$: 
	\begin{align}\label{equation: TraceIntegralII}
	    {\rm tr}_{D}(T)&=\frac{1}{\omega_D}\int_{\partial B^D_\ell}{\rm d^D x}\:\mathfrak{d}(x) T\left(\frac{x}{||x||},\frac{x}{||x||}\right).
	\end{align}
	We have thus introduced a scale $\ell>0$ into the problem. The first step in approximating the trace involves discretising this result; to do so, we note that we may choose a set of points $U$ in $\mathbb{R}^D$ such that a sum of the integrand over these points constitutes a Riemann sum for the integral Eq.~\ref{equation: TraceIntegralII}. We shall call any such set a \textit{trace grid}. It is natural to assume that these points lie in the annulus $S_{\ell,r}^D$---introducing a secondary scale $r>0$ that determines the thickness of the shell---and that the points are regular in the sense that adjacent points have constant separation. Since we wish to use $S_{\ell,r}^D$ to approximate $\partial B_\ell^D$, we have
	\begin{align}
	    r\ll \ell,
	\end{align}
	while the distance between points is essentially determined by the size $|U|$ of the set $U$. In particular if we let $M\in \mathbb{N}$ be such that
	\begin{align}
	    |U|=2M^{D-1},
	\end{align}
	the angular difference between adjacent points of $U$ will be $\mathcal{O}(M^{-1})$. The metric distance between adjacent points is consequently $\mathcal{O}((\ell+r)M^{-1})=\mathcal{O}(\ell M^{-1})$. We also find that the discretisation error is of order $\mathcal{O}(M^{-1})$.
	
	We now turn to the intrinsic characterisation of the trace approximation above. The idea is that for any $u\in G$, we may recursively construct a set $\tilde{U}\subseteq S_{\ell,r}^G(u)\subseteq G$ such that $\iota(\tilde{U})$ is almost $\exp_p(U)$, $p=\iota(u)$, for some trace grid $U\subseteq T_p\mathcal{M}$. When we say that $\iota(\tilde{U})$ is almost $\exp_p(U)$ this means there is a bijection $u\mapsto \tilde{u}$ such that $\rho_{\mathcal{M}}(\exp_p(u),\iota(\tilde{u}))<\varepsilon$. In fact we require slightly more since in the specification of the set $\tilde{U}$ we also need to be able to assign angular coordinates to elements of $\tilde{U}$ in such a way that the spherical volume element function $\mathfrak{d}$ can be approximated on $\tilde{U}$. Since $\iota(G)$ is, by assumption, a $\varepsilon$-net in $\mathcal{M}$ it is at some level obvious that we may find a set $\tilde{U}$ and a (surjective) mapping $u\mapsto \tilde{u}$ such that $\rho_{\mathcal{M}}(\exp_p(u),\iota(\tilde{u}))<\varepsilon$ for a trace grid $U$. It is a little harder, however, to ensure that this mapping is a bijection: the point is that while the exponential map is a radial isometry it is not a full isometry. Hence even though $\exp_p(S_{\ell,r}^D)=S_{\ell,r}^{\mathcal{M}}(p)$, we may find points in $U$ moved closer together by a distance determined by the distortion of the exponential map. This distortion can be characterised more or less precisely for manifolds with bounded sectional curvature:
	\begin{align}\label{equation: DistortionExponential}
	    {\rm dis}(\exp_p|B_R^D)=\mathcal{O}(R^3)
	\end{align}
	for geodesic balls $B_R^D$ with radius $0<R\ll 1$. In particular this means that the points in $U$ are shifted by a distance $\mathcal{O}((\ell+r)^3)=\mathcal{O}(\ell^3)$ and adjacent points in $\exp_p(U)$ can be as close as $\mathcal{O}(\ell M^{-1}-\ell^3)$. Thus if $\ell M^{-1}-\ell^3\sim \varepsilon$ we may have several points of $\exp_p(U)$ nearby the \textit{same} point of $\iota(G)$; if this occurs sufficiently often then the sum over points in $G$ matched to $\exp_p(U)$ will not approximate the sum over $U$ after all and we require:
	\begin{align}
	    (\varepsilon +\ell^3)M\ll \ell. 
	\end{align}
	This constraint is essentially a constraint on $M$: $M$ must be large since the error arising from treating Eq.~\ref{equation: TraceIntegralII} via a Riemann sum goes with the inverse of $M$. However if $M$ is too large, the points in $U$ become too close together to guarantee that a sum over matching points in $G$ will in fact approximate the Riemann sum of ${\rm tr}_D$.
	
	The recursive construction of $\tilde{U}$ itself ultimately boils down to being able to specify a notion of \textit{relative angle} between points of $G$, essentially via the cosine rule. For sufficiently nearby points we have a precise estimate of the error of this assignment when nearly isometrically imbedded in $\mathcal{M}$; thus if this error is much smaller than the angular separation between adjacent points in $U$, we can reliably pick out elements of $G$ which `should be' adjacent to points $v\in \tilde{U}$ if $\tilde{U}$ is almost a set . Such points will exist if $G$ approximates $\mathcal{M}$ since such points exist in $\mathcal{M}$. Note, however, that the recursive construction of $\tilde{U}$ proceeds intrinsically and as such it should be specified in such a way that the recursion terminates in all graphs; however the recursion will obviously fail to produce a set $\tilde{U}$ satisfying the above properties if $G$ is not Gromov-Hausdorff close to $\mathcal{M}$ since, for instance, there is no need for the points that `should be' adjacent to $v\in \tilde{U}$ to exist.
	
	Finally we note that a similar procedure takes place in the case of the Ollivier curvature: the problem reduces to showing that the uniform measure on the (discrete) set $\iota(B_{\delta}^G(u))$ approximates the uniform measure on the ball $B_{\delta}^{\mathcal{M}}(\iota(u))$. Showing that this is in fact the case essentially involves showing that the discrete measure is \textit{nearly} a uniform discrete measure on an even grid of points in Euclidean space; then known Euclidean semidiscrete optimal transport problems can show that this is near the uniform measure in $\mathbb{R}^D$ which is near to the relevant uniform measure in $\mathcal{M}$. Translating between $\mathcal{M}$ and $\mathbb{R}^D$ involves the use of the exponential map which introduces an error $\delta^3$ according to its distortion on the $\delta$-ball. Also, as in the case of the trace grid above, the evenly spaced points in Euclidean space must be sufficiently separated that they remain identifiable after a distortion of $\varepsilon+\delta^3$. The details are discussed in appendix \ref{appendix: Semidiscrete}.
    \subsection{Taking the Trace}
    As mentioned above, rigorous statements and proofs of steps $1$ and $3$ have been given in the appendices. Because several notions introduced in the process of taking the trace are required to define the discrete action we shall consider taking the trace in more detail here. Proofs of all theorems etc. in this section are found in appendix \ref{appendix: Trace}.
    
    In our informal outline above, we proceeded by discretising the Euclidean integral representation of the trace \ref{equation: TraceIntegral}. The idea is to introduce a set of points in $T_p\mathcal{M}$ in such a way that summing over these points gives a Riemann sum for the integral \ref{equation: TraceIntegralII}. More precisely we define Euclidean \textit{trace grids} as follows:
\begin{definition}
	Let $M\in \mathbb{N}$ be a (large) positive integer, and consider the multi-index $\boldsymbol{m}=(m_1,...,m_{D-1})$ where $m_1,...,m_{D-2}\in \set{1,...,M}$ and $m_{D-1}\in \set{1,...,2M}$. For any $\ell,\:r>0$ with $r\ll \ell$, a Euclidean \textit{trace grid} on $2M^{D-1}$ points in the annulus $S_{\ell,r}^D$ is any set $U=\set{(t_\textbf{m},u_\textbf{m})}$ indexed by the multi-index $\textbf{m}$ such that
	\begin{align}
		u_\textbf{m}=\left(\frac{\pi m_1}{M},\cdots ,\frac{\pi m_{D-1}}{M}\right)
	\end{align}
	in spherical coordinates and $\ell-r<t_\textbf{m}<\ell+r$ for all $\textbf{m}$ (up to a rigid rotation of the entire grid). For any trace grid $U$ and any bilinear form $T:\mathbb{R}^D~\times~ \mathbb{R}^D~\rightarrow~\mathbb{R}$ we define:
	\begin{align}
		\text{ tr}_U(T)=\frac{\pi^{D-1}}{\omega_DM^{D-1}}\sum_\textbf{m} \mathfrak{d}(u_{\textbf{m}})T(u_{\textbf{m}},u_{\textbf{m}}).
	\end{align}
\end{definition}
The point of this definition is that it immediately yields the following:
\begin{align}
	\left|\text{ tr}_{D}(T)-\text{ tr}_{U}(T)\right|=\mathcal{O}(M^{-1}),
\end{align}
i.e. as long as $M$ is large we may approximate the trace by summing over the points in the trace grid.

The aim is to characterise $\text{ tr}_U$ via a sum over some set in $G$. This involves (a) intrinsically specifying a set $\tilde{U}\subseteq G$ such that $\iota(\tilde{U})$ is \textit{almost} $\exp_p(U)$ for some Euclidean trace grid $U\subseteq T_p\mathcal{M}$ and (b) finding a way of characterising the summand $\mathfrak{d}(u_{\textbf{m}})T(u_{\textbf{m}},u_{\textbf{m}})$ for each multi-index $\textbf{m}$ via some function on $G$ that can be specified without reference to $\mathcal{M}$. For present purposes, we shall assume that we can in fact do this for the bilinear form $T$---this is the substance of the demonstration that the Ollivier curvature may be approximated---and we need only specify the spherical volume element function $\mathfrak{d}$. It turns out that both of these tasks can be achieved at once: the idea is that for each $u\in G$, we may recursively construct a mapping $\theta_u:S_{\ell,r}^G(u)\rightarrow U\cup \set{0}$ where $U$ is a Euclidean trace grid, such that if we take $\tilde{U}=\theta_u^{-1}(U)\subseteq S_{\ell,r}^G(u)$ we have that $\theta_u|\tilde{U}$ is an injection and $\exp_p(U)$ approximates $\iota(\tilde{U})$. By construction, then, $\tilde{U}$ resolves point (a) and since $\theta_u$ obviously assigns each $u\in \tilde{U}$ an angular coordinate (the coordinate of the corresponding point in $U$), we see that the summand of $\text{ tr}_U$ can also be well approximated on $\tilde{U}$.

How does one construct the mapping $\theta_u$? We begin by noting that given any root point in a Euclidean trace grid $U$---which by rotational invariance of the trace we may identify as $u_{(1,...,1)}\in U$---we can recursively obtain the entire trace grid by successively moving to adjacent points in the grid. In this sense it is sufficient to be able to identify the analogue of adjacent points in the set $\tilde{U}$. It turns out that this may be achieved given knowledge of the \textit{relative angle} between points in $G$. More precisely we define the following:
\begin{definition}
	For any $u\in G$ and for suitable $\ell,\:r>0$ and $v_1,\:v_2\in S_{\ell,r}^G(u)$, we define the \textit{relative angle} between $v_1$ and $v_2$ with respect to $u$ by
	\begin{align}
		\theta_u^G(v_1,v_2)=\arccos\left(\frac{a^2+b^2-c^2}{2ab}\right),
	\end{align}
	where $a=\rho_G(u,v_1)$, $b=\rho_G(u,v_2)$ and $c=\rho_G(v_1,v_2)$.
\end{definition}
We have a precise estimate for the error of the relative angle when $G$ is Gromov-Hausdorff close to $\mathcal{M}$:
\begin{lemma}\label{lemma: AngularErrorGHProximal}
	Let $p=\iota(u)\in \mathcal{M}$ for $u\in G$ and consider $v_1,\:v_2\in S_{\ell,r}^G(u)$ with $r,\:\ell>0$ and $r\ll \ell$. For $\ell$ sufficiently small, we have
	\begin{widetext}
	\begin{align}
		\left|\theta_D(V_1,V_2)-\theta_u^G(v_1,v_2)\right|=\left|\theta_D(V_1,V_2)-\theta_p^{\mathcal{M}}(V_1,V_2)\right|=\left|\theta_p^{\mathcal{M}}(V_1,V_2)-\theta_u^G(v_1,v_2)\right|=\mathcal{O}(\ell^2),
	\end{align}
	\end{widetext}
	where $V_k=\exp_{p}^{-1}(\iota(v_k))$ for $k\in \set{1,2}$ and $\theta_D(x,y)$ is the Euclidean angle between any vectors $x,\:y\in \mathbb{R}^D$.
\end{lemma}
We now turn to our recursive construction of $\theta_u$. Essentially we wish to prove the following:
\begin{theorem}\label{theorem: AngleAlgorithm}
	Let $G$ be a graph and let $\ell,\:r>0$ with $r\ll \ell$. For any positive integer $M$ such that
	\begin{align}\label{constraint: AngleTrace}
		M\ell^2\ll 1,
	\end{align}
	and for each $u\in G$, we can recursively construct a mapping $\theta_u:S_{\ell,r}^G(u)\rightarrow \mathbb{R}^{D}$ satisfying the following properties: 
	\begin{enumerate}
		\item $\theta_u(S_{\ell,r}^G(u))=U\cup \set{0}$ where $U\subseteq S_{\ell,r}^D$ is a Euclidean trace grid on $2M^{D-1}$ points. 
		\item If $\rho_{GH}(\mathcal{M},G)=\varepsilon/2$ for some $\varepsilon\ll r\ll \ell$ such that
		\begin{align}\label{constraint: EpsilonTrace}
			\varepsilon\ll \ell^3  && 2M^{D-1}\leq \mathcal{O}\left(\frac{r\ell^{D-1}}{\varepsilon^D}\right), 
		\end{align}
		then $|\theta_u^{-1}(U)|=2M^{D-1}$ and for any $\varepsilon$-isometry $\iota:G\rightarrow\mathcal{M}$ we have 
		\begin{align}
			\rho_{\mathcal{M}}(\exp_{\iota(u)}(\theta_u(v)),\iota(v))=\mathcal{O}(\ell^3)
		\end{align}
		for all $v\in \theta_u^{-1}(U)$. 
	\end{enumerate}
\end{theorem} 
\begin{definition}
	The function $\theta_u$ defined in theorem \ref{theorem: AngleAlgorithm} above is called the \textit{angular assignment} at $u\in G$ on $n=2M^{D-1}$ vertices. We then define the \textit{trace} of a function $f:S_{\ell,r}^G(u)\rightarrow\mathbb{R}$ at $u$ via the expression
	\begin{align}
		\text{tr}_G^u(f)= \frac{\pi^{D-1}}{\omega_D M^{D-1}}\sum_{v\in S_{\ell,r}^G(u)}\mathfrak{d}(\theta_u(v))f(v),
	\end{align}
	where the sum is over all multi
\end{definition}
The point of this definition is ultimately the following which shows that if we have a function on the graph that approximates a bilinear form on a manifold then the traces agree up to small errors:
\begin{corollary}\label{corollary: TraceApproximation}
	Let $\mathcal{M}$ be a Riemannian manifold and let $T$ be a bilinear form at $p\in \mathcal{M}$. Let $G$ be a graph and $\iota:G\hookrightarrow \mathcal{M}$ an $\varepsilon$-isometry with $p=\iota(u)$ for some $u\in G$. Let $\ell,\:r$ and $M$ be as in theorem \ref{theorem: AngleAlgorithm}, i.e. we have constraints \ref{constraint: AngleTrace} and \ref{constraint: EpsilonTrace}. Finally suppose there is a mapping $f:S_{\ell,r}^G(u)\rightarrow \mathbb{R}$ such that
	\begin{align}
		|f(v)-T(\theta_u(v),\theta_u(v))|=\mathcal{O}(\sigma)
	\end{align}
	for all $v\in \theta_u^{-1}(U)$. Then
	\begin{align}
		\left|\text{tr}_D(T)-\text{ tr}_G^u(f)\right|=\mathcal{O}(\max(M^{-1},\sigma)).
	\end{align}
\end{corollary}
\subsection{The Main Theorem}
In this section we prove our main theorem, viz. there is a discrete Einstein-Hilbert action that converges to its counterpart on suitable sequences of graphs that converge to a compact Riemannian manifold.
We first define the discrete Einstein-Hilbert action.
\begin{definition}
    Let $G$ be a graph with $N$ vertices, let $\delta,\:\ell,\:r>0$ be real numbers and let $M$ be an integer smaller than $N$. Then we define the \textit{discrete Einstein-Hilbert action} for parameters $\delta$, $\ell$, $r$ and $M$ via the assignment:
    \begin{widetext}
	\begin{align}\label{equation: DiscreteEHAction}
		\mathcal{A}_{DEH}(G;\delta,\ell,r,M)=\frac{1}{N\delta^2}\sum_{u\in G}\frac{4\pi^{D-1}(D+2)}{M^{D-1}}\sum_{v\in S^G_{\ell,r}(u)}\mathfrak{d}\left(\theta_u(v)\right)\kappa_G^\delta(u,v),
	\end{align}
	\end{widetext}
	where for each $u\in G$, $\theta_u$ is an angular assignment on $2M^{D-1}$ points at $u$. For any given compact Riemannian manifold $\mathcal{M}$, the \textit{error} in the discrete Einstein-Hilbert action is thus defined:
	\begin{align}
	    \delta\mathcal{A}_{DEH}(\delta,\ell,r,M)&=\left|\mathcal{A}_{DEH}(G;\delta,\ell,r,M)-\mathcal{A}_{EH}(\mathcal{M}))\right|
	\end{align}
\end{definition}
\begin{theorem}\label{theorem: MainTheorem}
	Let $G$ be a graph with $N$ vertices, let $\mathcal{M}$ be a compact Riemannian $D$-manifold, $D>1$, and let $\iota:G\rightarrow \mathcal{M}$ be an $\varepsilon(N)$-isometry such that the balls $\set{B_{\varepsilon(1-\alpha)}^{\mathcal{M}}(\iota(u))}_{u\in G}$ are pairwise disjoint for some $\alpha>0$ with $\alpha\ll 1$. Let $R$ be the Ricci scalar of $\mathcal{M}$ and suppose that the positive constants $K$ and $\tilde{K}$ are such that $R\vert B_{\varepsilon}^{\mathcal{M}}(\iota(u))$ is $K$-Lipschitz for each $u\in G$, the uniform norm $\vert \vert R\vert \vert_{\infty}\leq \tilde{K}$ and
	\begin{align}\label{constraint: Manifold}
		K\varepsilon\ll 1\ll K && \tilde{K}\alpha \ll 1.
	\end{align}
	Also let:
	\begin{subequations}\label{equation: Scales}
	\begin{align}
		\varepsilon(N)&=N^{-\frac{1}{D}} & \delta(N)&= N^{-a} & \ell(N)&= N^{-b} \\ r(N)&= N^{-c} & M(N)&=N^{d}.
	\end{align}
	\end{subequations}
	For $N$ sufficiently large, any choice of numbers $a,\:b,\:c,\:d>0$ such that
	\begin{subequations}\label{constraint: Scales}
		\begin{align}
			b<a,\:c<\frac{1}{D}\\
			(3a+b)D<1<4aD \\
			\frac{1}{2}d<b\leq \frac{1-c}{D-1}-d
		\end{align}
	\end{subequations}
	ensures that
	\begin{align}\label{equation: ActionError}
		\delta\mathcal{A}_{DEH}(\delta,\ell,r,M)=\mathcal{O}\left(K\varepsilon,\tilde{K}\alpha,\sigma\right)
	\end{align}
	where
	\begin{align}
	    \sigma=\max\left(N^{-a(3+2b)},N^{-\frac{1-(c+(D-1)b)}{D-1}}\right).
	\end{align}
	Under these conditions, $\sigma$ is small $(\ll 1)$ for sufficiently large $N$.
\end{theorem}
\begin{remark}
	We have two types of constraint in the statement of the above theorem: the constraints \ref{constraint: Scales} are basically constraints on the relevant scales of the problem required in order to ensure convergence at various different levels. There is a real possibility \textit{a priori} that these constraints are incompatible and part of the proof will precisely be an attempt to show that the above constraints are in fact compatible. The second type of constraint are given by the inequalities \ref{constraint: Manifold} and are required due to their role in lemma \ref{lemma: ApproximatingIntegrals}. Naively, these constraints are essentially presented as a constraint on the type of limiting manifold for which the above discrete Einstein-Hilbert action can be shown to converge. Certainly the quantities $K$ and $\tilde{K}$ are defined as bounds on the Ricci scalar of the manifold. However, it is perhaps better to interpret these as restrictions on the types of convergent sequence for which the discrete Einstein-Hilbert action will converge.
\end{remark}
\begin{remark}
	Note that to make the action $\mathcal{A}_{DEH}$ intrinsic, we need to be able to specify $D$ independently of the limiting manifold; noting that involved in the above claim is the fact that small balls at each point converge to small balls in the manifold, we see that the intrinsic Hausdorff dimension---defined as the power relating the radius of a ball to its volume---also converges and we may characterise $D$ intrinsically in this manner.
\end{remark}
As an immediate corollary to the above theorem we thus have:
\begin{corollary}
	For any space of (weighted) graphs $\Omega$, there is a discrete Einstein-Hilbert action $\mathcal{A}_{DEH}:\Omega\rightarrow \mathbb{R}$ given by Eq.~\ref{equation: DiscreteEHAction} such that $\mathcal{A}_{DEH}(\omega_N)\rightarrow \mathcal{A}_{EH}(\mathcal{M})$
	for any suitable sequence of graphs $\set{\omega_N}_{N\in \mathbb{N}^+}\subseteq \Omega$ such that $|\omega_N|=N$ and $\omega_N\rightarrow \mathcal{M}$ sufficiently rapidly in the sense of Gromov-Hausdorff. 
\end{corollary}
\subsection{Klein-Gordon Fields}
In this section we briefly summarise our method of dealing with Klein-Gordon fields, defined as scalar fields which locally satisfy the Euler-Lagrange equations associated with the following Klein-Gordon Lagrangian:
\begin{align}
	\mathcal{L}_{KG}(\varphi,\text{d}\varphi)=\frac{1}{2}\text{ tr}(\text{d}\varphi \otimes \text{d}\varphi)+\frac{1}{2}m\varphi^2.
\end{align}
The point is that previously derived results allow us to approximate the Klein-Gordon action
\begin{align}
	\mathcal{A}_{KG}=\int_{\mathcal{M}}\text{d}\text{ vol}_{\mathcal{M}}(x)\mathcal{L}_{KG}(\varphi(x),\text{d}\varphi(x))
\end{align}
If we have points $a,\:b\in \mathcal{M}$ such that $\rho_{\mathcal{M}}(a,b)<\delta$ for $\delta$ sufficiently small, then we have a unit tangent vector $V$ at $a$ for the unique geodesic connecting $a$ and $b$; by definition this satisfies:
\begin{align}
	\left|\text{d}\varphi(V)-\frac{\varphi(b)-\varphi(a)}{\rho_{\mathcal{M}}(a,b)}\right|=\mathcal{O}(\delta).
\end{align}
The following then follows immediately from an application of lemma \ref{lemma: ApproximatingIntegrals} and corollary \ref{corollary: TraceApproximation}:
\begin{proposition}\label{proposition: KleinGordon}
	Let $\omega_n\rightarrow \mathcal{M}$ in the sense of Gromov-Hausdorff as above and let $f_n:\omega_n\rightarrow \mathbb{R}$ be a sequence of functions that converges pointwise to a smooth function $\varphi:\mathcal{M}\rightarrow\mathbb{R}$. This means that for any sequence of $\varepsilon_n$-nets $\iota_n:\omega_n\rightarrow\mathcal{M}$ such that $\varepsilon_n\rightarrow 0$, and any sequence of points $u_n\in \omega_n$ such that $\iota_n(u_n)\rightarrow p\in \mathcal{M}$, we have $f_n(u_n)\rightarrow \varphi(p)$. If $\varphi$ is a Klein-Gordon field, i.e. if it extremises the Klein-Gordon action $\mathcal{A}_{KG}$, then there is a discrete Klein-Gordon action $\mathcal{A}_{DKG}$ such that $\mathcal{A}_{DKG}(f_n)\rightarrow \mathcal{A}_{KG}(\varphi)$.
\end{proposition}
Higher interaction terms can easily be incorporated, but it is not clear how one might go about approximating higher valence (vector/tensor) fields. The main point is that in order to take the derivative of such fields one needs to be able to approximate the smooth Levi-Civita connection; our entire approach, however, has been based around the idea that this can be dispensed with since the phenomenal quantities of conventional Einstein-Gravity can be determined via the Ricci tensor alone, something which can be determined by the metric and measure theoretic structure, allowing us to dispense with the full smooth structure. Presumably, given knowledge of the connection---a strictly stronger requirement since it would permit us to reconstruct the full Riemann curvature tensor---we would be much closer to reproducing the full smooth structure of the Riemannian manifold in question, and it seems unlikely that we have many prospects in this direction. From this perspective, the fact that we can nonetheless obtain convergent Klein-Gordon (scalar) fields appears to be something of an accident.

	\section{Conclusions}
	Let us briefly make some remarks about the significance of this paper: this contribution is part of a programme pursued by the authors of combinatorial quantum gravity, in which we have attempted to utilise new mathematical insights into the characterisation of course Ricci curvature to shed new light on certain problems in Euclidean quantum gravity. Thus far \cite{Trugenberger_CombQG,KellyEtAl_Circle,KellyEtAl}, we have investigated a suggestive random graph model which displays some particularly clear signatures of emergent geometric structure including a classical phase dominated by near-lattice configurations and good evidence of a continuous phase transition. It is not entirely clear how that model relates to Euclidean quantum gravity proper, however, because the action adopted is at best a formal discretisation of the continuum Einstein-Hilbert action which agrees on the Ricci-flat sector. Indeed Akara-pipattana, Chotibut and Evnin \cite{Akara-pipattanaChotibutEvnin-Emergence} have expressed reservations about the gravitational interpretation of our work. The result of the present paper is thus a necessary step in the programme of combinatorial quantum gravity, and we hope it addresses some of these reservations.
	
	Conceptually, the result shows that if a graph is like a manifold---in the precise sense that it is Gromov-Hausdorff proximal to a manifold---then the discrete Einstein-Hilbert action of the graph will be very nearly the continuum Einstein-Hilbert action of the manifold. By extension such a graph will contribute a similar amount to the partition function as its Gromov-Hausdorff proximal manifold configurations. The partition function itself of any model of random graphs will probably look rather different from that of naive Euclidean quantum gravity or---less naively---of Euclidan dynamical triangulations. As shown by Loisel and Romon \cite{LoiselRomon-OllCurvPolyhedral}, discrete notions of curvature often differ significantly even on the most simple configurations. Ollivier curvature, in particular, is unlikely to look like Regge calculus in many instances since it is sensitive to the precise structure of local clustering; for instance if we take a regular triangulation and regular quadrangulation of the same surface, an Ollivier curvature based action will give very different results for the two discrete configurations. The nature of the configuration spaces on which the models are defined are also potentially very different. We stress that from our perspective, significant differences from Regge calculus on discrete configurations is desirable: if Ollivier curvature base models are too similar to dynamical triangulations models we could hardly avoid pathologies like the branched polymer phase.
	
	There is a slightly different way of looking at the main result of the present paper: we are now in a position to specify a statistical model of random graphs that is \textit{kinematically} consistent with Euclidean gravity. The hope is that the dynamics give rather different results from Euclidean dynamical triangulations. From this perspective, the present result is rather similar to the results of Benincasa and Dowker \cite{BenincasaDowker_ScalCurvCausSet} which gives an action on causal sets that agrees on sprinklings in Lorentzian manifolds. However, while it is not clear how a dynamical model of random graphs (or causal sets) can result in a random geometric graph (sprinkling) in some background manifold, it is quite possible to obtain (even spontaneously) graphs that are Gromov-Hausdorff proximal to particular manifolds as we showed in \cite{KellyEtAl_Circle}.
	
	Several difficulties remain, however. Firstly there is the problem of noncompact manifolds: the Gromov-Hausdorff metric is only defined on the space of (isometry classes of) compact metric spaces and the question of noncompact limits immediately arises. In fact there is a sense of \textit{pointed Gromov-Hausdorff} convergence for \textit{locally} compact geodesic spaces, where one considers Gromov-Hausdorff convergence with respect to some reference family of compact subsets of the space. One potential issue with this approach is that pointed Gromov-Hausdorff convergence depends on a choice of reference point, leading to a loss of gauge invariance. Another issue is that the Gromov-Hausdorff distance is hard to compute \cite{Schmiedl-Discrete}, so even knowing that a scaling limit exists may be somewhat uninformative if the structure of the scaling limit is not known.
	
	There is an interesting question about connections to causal models: all of our work has been firmly concerned with Euclidean gravity, and due to the significance of metric structure for the specification of Ollivier curvature it is unclear to the authors how and if the present results extend to Lorentzian contexts. It is worth noting that Gorard \cite{Gorard} has suggested that the Ollivier curvature extends to Lorentzian manifolds with little extra work, while there has been some interesting work on optimal transport in Lorentzian manifolds \cite{McCann,MondinoSuhr,EcksteinMiller-Causality}. Another interesting question is on the existence of coarse analogues of the Einstein field equations: again, Gorard has suggested that discrete Einstein-Field equations may be derived using a kind of formal analogue of Chapman-Enskog theory, though we confess that we have yet to fully understand this argument. Conceptually, the role of the field equations is to provide a `local' test of extremality in the sense that they allow one to check extremality without knowledge of any other configurations. As such, while undoubtedly convenient, it is not clear that the existence of field equations has any fundamental significance. In summary, we believe that both the problem of Lorentzian coarse curvature and the problem of rough Einstein field equations are worth pursuing further, but we have little of interest to say about them at present.
	\section*{Acknowledgements}
    C.K. would like to thank Pim van der Hoorn for explaining various aspects of the references \cite{Hoorn_Convergence,HoornEtAl_OllCurvConv} and Timothy Budd for pointing out an error in one of the lemmas in an earlier version of this text; he also acknowledges studentship funding from EPSRC under the grant number EP/L015110/1. F.B. acknowledges funding from EPSRC (UK) and the Max Planck Society for the Advancement of Science (Germany).
	\appendix
	\section{Approximating Integrals}\label{appendix: Integral}
    In this appendix we formalise the discussion in section \ref{subsection: InformalOutline} on the approximation of integrals of functions on $\mathcal{M}$:
    \begin{lemma}\label{lemma: ApproximatingIntegrals}
	Let $\alpha>0$ be a constant such that $\alpha\ll 1$ and such that the family $\set{B_{\varepsilon(1-\alpha)}^{\mathcal{M}}(\iota(u))}_{u\in G}$ is pairwise disjoint. Let $f:\mathcal{M}\rightarrow\mathbb{R}$ be a function which is $K$-locally Lipschitz in the balls $\set{B_{\varepsilon}^{\mathcal{M}}(\iota(u))}_{u\in G}$ with bounded uniform norm $\vert\vert f\vert\vert_{\infty}\leq \tilde{K}$, where the constants $K$ and $\tilde{K}$ satisfy
	\begin{align}\label{Constraint: Integral}
		K\varepsilon\ll 1\ll K && \tilde{K}\alpha\ll 1,
	\end{align}
	and let $g:G\rightarrow\mathbb{R}$ be a function such that
	\begin{align}\label{equation: IntegrandError}
		|f(\iota(u))-g(u)|=\mathcal{O}(\sigma)
	\end{align}
	for some $\sigma>0$. Then if
	\begin{align}\label{equation: IntegralScaling}
		\varepsilon=N^{-\frac{1}{D}},
	\end{align}
	we have
	\begin{widetext}
	\begin{align}\label{equation: IntegralApproximation}
		\left|\int_{\mathcal{M}} \text{d vol}_{\mathcal{M}}(x)f(x)-\frac{\omega_D}{N}\sum_{u\in G}g(u)\right|=\mathcal{O}\left(\max(K\varepsilon,\tilde{K}\alpha,\sigma)\right).
	\end{align}
	\end{widetext}
\end{lemma}
\begin{proof}
	Recall that any measurable function $f$ can be expressed as a sum $f=f_+-f_-$ where $f_{\pm}$ are positive measurable functions; in this way $\int f=\int f_+-\int f_-$ and to show that we may approximate $\int f$ it is sufficient to show we may approximate $\int f$ for $f$ positive. Thus let us assume that $f:\mathcal{M}\rightarrow \mathbb{R}$ is positive without loss of generality.
	
	Now recall that $\iota:G\hookrightarrow \mathcal{M}$ is an $\varepsilon$-isometry, $\iota(G)$ is an $\varepsilon$-net in $\mathcal{M}$ and the balls $\set{B_\varepsilon^{\mathcal{M}}(\iota(u))}_{u\in G}$ form an open cover of $\mathcal{M}$. We may choose a partition of unity $\set{\rho_u}_{u\in G}$ subordinate to $\set{B_\varepsilon^{\mathcal{M}}(\iota(u))}_{u\in G}$ such that $\rho_u(u)=1$ for all $u\in G$. Since $\rho_u$ takes values in $[0,1]$ for all $u\in G$ we note that $\rho_u f\leq f$ and
	\begin{widetext}
	\begin{align}\label{inequality: IntegralApproxI}
		\int_{\mathcal{M}}\text{d vol}_{\mathcal{M}}(x)f(x)&=\sum_{u\in G}\int_{B_\varepsilon^{\mathcal{M}}(\iota(u))} \text{d vol}_{\mathcal{M}}(x)\rho_u(x)f(x)\nonumber\\
		&\leq \sum_{u\in G}\int_{B_\varepsilon^{\mathcal{M}}(\iota(u))}\text{d vol}_{\mathcal{M}}(x)f(x)\nonumber\\
		&\leq \sum_{u\in G}\int_{B_\varepsilon^{\mathcal{M}}(\iota(u))}\text{d vol}_{\mathcal{M}}(x)(f(\iota(u))+\vert f(x)-f(\iota(u))\vert)\nonumber\\
		&\leq \sum_{u\in G}\int_{B_\varepsilon^{\mathcal{M}}(\iota(u))}\text{d vol}_{\mathcal{M}}(x)(f(\iota(u))+K\rho_{\mathcal{M}}(\iota(u),x))\nonumber\\
		&\leq \sum_{u\in G}\int_{B_\varepsilon^{\mathcal{M}}(\iota(u))}\text{d vol}_{\mathcal{M}}(x)(f(\iota(u))+K\varepsilon)\nonumber\\
		&=\sum_{u\in G}(f(\iota(u))+K\varepsilon)\text{vol}_{\mathcal{M}}(B_\varepsilon^{\mathcal{M}}(\iota(u)))\nonumber\\
		&=\omega_D\varepsilon^D\sum_{u\in G}(f(\iota(u))+K\varepsilon)(1+\mathcal{O}(\varepsilon))\nonumber\\
		&=\frac{\omega_D}{N}\sum_{u\in G}f(\iota(u))+\mathcal{O}(K\varepsilon).
	\end{align}
	\end{widetext}
	We have used the Lipschitz continuity of $f\vert B_{\varepsilon}^{\mathcal{M}}(\iota(u))$, $u\in G$, in moving to the fourth line, the fact that $\text{ vol}(B_\varepsilon^{\mathcal{M}}(q))=\varepsilon^D\omega_D(1+\mathcal{O}(\varepsilon))$ for all $q\in \mathcal{M}$ in moving to the sixth line, and the identification $N=\varepsilon^D$ in moving to the final line. 
	
	On the other hand, using the fact that the family $\set{B_{r}^{\mathcal{M}}(\iota(u))}_{u\in G}$ is pairwise disjoint for $r=\varepsilon(1-\alpha)$ we note that there is a partition of unity $\set{\rho_u}_{u\in G}$ subordinate to the cover $\set{B_{\varepsilon}^{\mathcal{M}}(\iota(u))}_{u\in G}$ such that $\rho_u\vert B_{r}^{\mathcal{M}}(\iota(u))=1$. Then again by the definition of the integral we have
	\begin{align}
		\int_{\mathcal{M}}\text{d vol}_{\mathcal{M}}(x)f(x)&=\sum_{u\in G}\int_{B_\varepsilon^{\mathcal{M}}(\iota(u))}\text{d vol}_{\mathcal{M}}(x)\rho_u(x)f(x)\nonumber\\
		&\geq \sum_{u\in G}\int_{B_r^{\mathcal{M}}(\iota(u))}\text{d vol}_{\mathcal{M}}(x)f(x)\nonumber.
	\end{align}
	Note also that for every $x\in B_\varepsilon^{\mathcal{M}}(\iota(u))$, we have:
	\begin{align}
		f(\iota(u))-K\varepsilon\leq f(x).\nonumber
	\end{align}
	This is trivial if $f(\iota(u))\leq f(x)$ while if $f(\iota(u))>f(x)$ we have by Lipschitz continuity
	\begin{align}
		f(\iota(u))-f(x)=\vert f(\iota(u))-f(x)\vert \leq K\rho_{\mathcal{M}}(\iota(u),x)<K\varepsilon\nonumber.
	\end{align}
	Thus
	\begin{widetext}
	\begin{align}
		\int_{\mathcal{M}}\text{d vol}_{\mathcal{M}}(x)f(x)&\geq \sum_{u\in G}\int_{B_r^{\mathcal{M}}(\iota(u))}\text{d vol}_{\mathcal{M}}(x)f(x)\nonumber\\
		&\geq \sum_{u\in G}\int_{B_r^{\mathcal{M}}(\iota(u))}\text{d vol}_{\mathcal{M}}(x)(f(\iota(u))-K\varepsilon)\nonumber\\
		&=\sum_{u\in G}(f(\iota(u))-K\varepsilon)\text{vol}_{\mathcal{M}}(B_r^{\mathcal{M}}(\iota(u)))\nonumber\\
		&=\omega_D\sum_{u\in G}(f(\iota(u))-K\varepsilon)r^D(1-\mathcal{O}(r))\nonumber\\
		&=\omega_D\varepsilon^D\sum_{u\in G}(f(\iota(u))-K\varepsilon)(1-\alpha)^D(1-\mathcal{O}(\varepsilon))\nonumber\\
		&=\frac{\omega_D}{N}\sum_{u\in G}(f(\iota(u))-K\varepsilon)(1-\mathcal{O}(\alpha))(1-\mathcal{O}(\varepsilon))\nonumber.
	\end{align}
	Multiplying out the right-hand side and using the fact that
	\begin{align}
		\mathcal{O}(\alpha)\frac{\omega_D}{N}\sum_{u\in G}f(\iota(u))\leq \mathcal{O}(\alpha)\frac{\omega_D}{N}\sum_{u\in G}\tilde{K}=\mathcal{O}(\tilde{K}\alpha)\nonumber
	\end{align}
	\end{widetext}
	means that we finally have
	\begin{align}
		\int_{\mathcal{M}}\text{d vol}_{\mathcal{M}}(x)f(x)&\geq \frac{\omega_D}{N}\sum_{u\in G}f(\iota(u))-\mathcal{O}(\max(K\varepsilon,\tilde{K}\alpha)).\nonumber
	\end{align}
	Combining this inequality with inequality \ref{inequality: IntegralApproxI} gives
	\begin{align}
		\left\vert \int_{\mathcal{M}}\text{d vol}_{\mathcal{M}}(x)f(x)- \frac{\omega_D}{N}\sum_{u\in G}f(\iota(u))\right\vert =\mathcal{O}(\max(K\varepsilon,\tilde{K}\alpha)).\nonumber
	\end{align}
	But then by subadditivity we have
	\begin{widetext}
	\begin{align}
		\left\vert \int_{\mathcal{M}}\text{d vol}_{\mathcal{M}}(x)f(x)- \frac{\omega_D}{N}\sum_{u\in G}g(u)\right\vert &\leq \left\vert \int_{\mathcal{M}}\text{d vol}_{\mathcal{M}}(x)f(x)- \frac{\omega_D}{N}\sum_{u\in G}f(\iota(u))\right\vert\nonumber\\
		&\qquad +\left\vert \frac{\omega_D}{N}\sum_{u\in G}f(\iota(u))- \frac{\omega_D}{N}\sum_{u\in G}g(u)\right\vert\nonumber\\
		&=\mathcal{O}(\max(K\varepsilon,\tilde{K}\alpha))+\left\vert \frac{\omega_D}{N}\sum_{u\in G}(f(\iota(u))-g(u))\right\vert\nonumber\\
		&\leq \mathcal{O}(\max(K\varepsilon,\tilde{K}\alpha))+\frac{\omega_D}{N}\sum_{u\in G}\left\vert f(\iota(u))-g(u)\right\vert \nonumber\\
		&=\mathcal{O}(\max(K\varepsilon,\tilde{K}\alpha))+\frac{\omega_D}{N}\sum_{u\in G}\mathcal{O}(\sigma)\nonumber\\
		&=\mathcal{O}(\max(K\varepsilon,\tilde{K}\alpha))+\mathcal{O}(\sigma)\nonumber\\
		&=\mathcal{O}(\max(K\varepsilon,\tilde{K}\alpha,\sigma))\nonumber
	\end{align}
	\end{widetext}
	as required.
    \end{proof}
    \section{Trace Error Proofs}\label{appendix: Trace}
    \begin{proof}[Proof of Lemma \ref{lemma: AngularErrorGHProximal}]
	Let $p\in \mathcal{M}$ and let $\gamma_1$ and $\gamma_2$ be geodesics issuing from $p$, i.e. $\gamma_1(0)=\gamma_2(0)=p$. If $\theta$ is the angle between the geodesics $\gamma_1$ and $\gamma_2$ at $p$, i.e. the angle $\theta_p^{\mathcal{M}}((\dot{\gamma}_1)_0,(\dot{\gamma}_2)_0)$, then for $s,t>0$ sufficiently small we have the standard Taylor expansion
	\begin{align}
		\rho_{\mathcal{M}}^2(\gamma_1(s),\gamma_2(t))=s^2+t^2-2st\cos \theta +\mathcal{O}((s,t)^4)\nonumber
	\end{align}
	where $\mathcal{O}((s,t)^4)$ means that the error is fourth-order in products of $s$ and $t$. This follows from the Jacobi equation. Identifying $p=\iota(u)$, $\gamma_1(s)=\iota(v_1)$ and $\gamma_2(t)=\iota(v_2)$ for $u\in G$ and $v_1,\:v_2\in S_{\ell,r}^G(u)$, we may replace manifold distances by graph distances at the cost of an error of order $\varepsilon \ell$ since $\text{dis}(\iota)=\varepsilon$:
	\begin{align}
		\rho_G(v_1,v_2)^2=x^2+y^2-2xy\cos\theta+\mathcal{O}(\max(\varepsilon \ell,(s,t)^4))\nonumber,
	\end{align}
	where $x= \rho_G(u,v_1)$ and $y= \rho_G(u,v_2)$. Rearranging this expression gives
	\begin{align}
		\cos \theta=\cos \theta_u^G(v_1,v_2)+\mathcal{O}(\max(\varepsilon\ell^{-1},\ell^2)),\nonumber
	\end{align}
	if we note that $s,\:t=\mathcal{O}(\ell+r)=\mathcal{O}(\ell)$. But $\varepsilon\ll \delta^2\ell\ll\ell^3$ and $\varepsilon/\ell\ll\ell^2$, i.e. $\cos \theta=\cos \theta_u^G(v_1,v_2)+\mathcal{O}(\ell^2)$. The smoothness of $\arccos$ then implies that
	\begin{align}
		|\theta-\theta_u^G(v_1,v_2)|=\mathcal{O}(\ell^2).\nonumber
	\end{align}
	But we also have the asymptotic expansion
	\begin{align}
		\braket{V_1,V_2}_p=||V_1||_p\cdot ||V_2||_p\cos \theta +\mathcal{O}((||V_1||_p, ||V_2||_p)^4)\nonumber,
	\end{align}
	and by the definition of $\theta_p^{\mathcal{M}}(V_1,V_2)$ we have
	\begin{align}
		\cos \theta_p^{\mathcal{M}}(V_1,V_2)=\cos \theta+\mathcal{O}(\ell^2)\nonumber
	\end{align}
	where we have used $||V_1||_p=s$, $||V_2||_p=t$ and $s,\:t=\mathcal{O}(\ell)$. The smoothness of $\arccos$ then gives
	\begin{align}
		|\theta_p^{\mathcal{M}}(V_1,V_2)-\theta|=\mathcal{O}(\ell^2).\nonumber
	\end{align}
	Hence the required result follows from subaddditivity.
    \end{proof}
	To prove theorem \ref{theorem: AngleAlgorithm}, it will be helpful to have some terminology:
\begin{definition}
	Let $G$ be a graph and pick $\ell,\:r>0$ and the positive integer $M$ as above, i.e. $M\ell^2\ll 1$ and $r\ll \ell$. Fix some $u\in G$.
	\begin{enumerate}
		\item Two vertices $v_1,\:v_2\in S_{\ell,r}^G(u)$ are said to be \textit{$(M,\ell)$-adjacent} iff 
		\begin{align}
			\left|\theta_G^u(v_1,v_2)-\frac{\pi}{M}\right|=\mathcal{O}(\ell^2).
		\end{align}
		\item Let $v_1,\:v_2\in S_{\ell,r}^G(u)$ be $(M,\ell)$-adjacent. A vertex $v_3\in S_{\ell,r}^G(u)$ that is $(M,\ell)$-adjacent to $v_1$ is \textit{$(M,\ell)$-perpendicular} to the pair $(v_1,v_2)$ iff
		\begin{align}
			\left|\theta^G_{v_1}(v_2,v_3)-\frac{\pi}{2}\right|=\mathcal{O}\left(M\ell^2\right).
		\end{align}
		\item A sequence of points $v_0,...,v_n$ such that $v_k$ and $v_{k+1}$ are $(M,\ell)$-adjacent for all $k\in \set{1,...,n-1}$ are said to be \textit{in line} iff for all $k_1,\:k_2,\:k_3\in \set{1,...,n}$ such that $k_1<k_2<k_3$ we have
		\begin{align}
			\theta_G^u(v_{k_1},v_{k_3})=\theta_G^u(v_{k_1},v_{k_2})+\theta_G^u(v_{k_2},v_{k_3})+\mathcal{O}(\ell^2).
		\end{align}
	\end{enumerate}
\end{definition}
\begin{proof}[Proof of Theorem \ref{theorem: AngleAlgorithm}]
	Recall that any Euclidean trace grid on $2M^{D-1}$ points is given $U=\set{u_{\boldsymbol{m}}}$ for the multi-index $\boldsymbol{m}=(m_1,...,m_{D-1})$ where $m_1,...,m_{D-2}\in \set{1,...,M}$ and $m_{D-1}\in \set{1,...,2M}$. Thus it is sufficient to pick out a suitable set $V\subseteq S_{\ell,r}^G(u)$ such that $V=\set{v_{\boldsymbol{m}}}$ and then identify
	\begin{align}
		\theta_u(v)=\left\{\begin{array}{rl}
			u_{\boldsymbol{m}}, & v=v_{\boldsymbol{m}}\\
			0, & \text{otherwise} 
		\end{array}\right.\nonumber
	\end{align}
	The recursive construction of the set $V$ must satisfy the following properties if $\theta_u$ is to satisfy the required properties: the mapping $\boldsymbol{m}\rightarrow v_{\boldsymbol{m}}$ is injective and $\rho_{\mathcal{M}}(\exp_{\iota(u)}(u_{\boldsymbol{m}}),\iota(v_{\boldsymbol{m}}))=\mathcal{O}(\ell^3)$ whenever $G$ is Gromov-Hausdorff close to $\mathcal{M}$, where the precise meaning of this phrase is given by the conditions on $\varepsilon$ in the statement of the theorem.
	
	We now specify the recursive construction of $V$. First note that the set $U\cup \set{0}$ can be given a graph structure in the following manner: connect $0$ to every vertex in $U$ and connect two vertices in $U$ iff they are adjacent i.e. their angular separation is $\pi/M$. Thus the notions of $(M,\ell)$-adjacency etc. extend to this graph. Given $u_{(1,...,1)}$ note that we can obtain $U$ (up to a relabelling, or equivalently a rigid rotation about the axis of $u_{(1,...,1)}$) recursively in the following manner:
	\begin{enumerate}
		\item Pick $u_{(2,1,...,1)}$ as any element of $U$ that is $(M,0)$-adjacent to $u_{(1,...,1)}$. Assume that we have picked $u_{(1,...,2,...,1)}$ where the $2$ is in the $k$th position, for all $k\leq K<D-1$. Then $u_{(1,...,2,...,1)}$ with the $2$ in the $(K+1)$th position is an element of $U$ that is $(M,0)$-adjacent to $u_{(1,...,1)}$ and $(M,0)$-perpendicular to $u_{(1,...,2,...,1)}$ with $2$ in the $k$th position for all $k\in \set{1,...,K}$. 
		\item Now assume that we have obtained $u_{\boldsymbol{m}}$ for all multi-indices $\boldsymbol{m}=(m_1,...,m_{D-1})$ such that $\sum_{k=1}^{D-1}m_k-(D-1)=K$, $K\in \set{0,1,...}$. Let $\boldsymbol{m}=(m_1,...,m_{D-1})$ be a multi-index such that $\sum_{k=1}^{D-1}m_k-(D-1)=K+1$. Either $u_{\boldsymbol{m}}$ is uniquely specified as the unique vertex that is both $(M,0)$-perpendicular and $(M,0)$-adjacent to some family of points already specified, or it extends a sequence of already specified points that are in line. 
	\end{enumerate}
	We define $V$ by applying the above algorithm to points in $S_{\ell,r}^G(u)$, where $v_{(1,..,1)}$ can be chosen arbitrarily and $(M,0)$-adjacency and $(M,0)$-perpendicularity are replaced with $(M,\ell)$-adjacency etc. There is a caveat insofar as unlike in the case of the graph $U\cup \set{0}$ it is not clear that the required `adjacent' point will exist according to this algorithm; in such a situation the algorithm terminates early (alternatively we set $v_{\boldsymbol{m}}=v_{(1,...,1)}$) and $\theta_u$ does not have the desired properties (may not be well defined). 
	
	It is clear that the algorithm does not terminate early if $G$ is Gromov-Hausdorff close to a manifold $\mathcal{M}$ by lemma \ref{lemma: AngularErrorGHProximal}: the graph $U\cup \set{0}$ pushes forward to a distorted graph in $\mathcal{M}$ and so the required adjacent point exists in $S_{\ell,r}^{\mathcal{M}}(p)$ as long as we weaken $(M,0)$-adjacency etc to $(M,\ell)$-adjacency. But then we can pick points in $S_{\ell,r}^G(u)$ that are $\mathcal{O}(\varepsilon)$ close to the relevant points of $S_{\ell,r}^{\mathcal{M}}(p)$; we pick up an angular error of $\mathcal{O}(\varepsilon/\ell)$ due to the metric distortion of the near isometry $\iota$. Thus $\varepsilon\ll \ell^3$ ensures that this error is $\mathcal{O}(\ell^2)$ and thus that the algorithm does not terminate early. This is not sufficient to ensure that $\theta_u$ is well defined: the mapping $\boldsymbol{m}\mapsto v_{\boldsymbol{m}}$ must also be injective. This is guaranteed as long as we have
	\begin{align}
		\rho_{\mathcal{M}}(\exp_{\iota(u)}(u_{\boldsymbol{m}}),\iota(v_{\boldsymbol{m}}))\sim R\ll\frac{\ell}{M}-\ell^3\nonumber
	\end{align}
	for all $\boldsymbol{m}$: the quantity on the right-hand side controls (up to scale factors) the distance between adjacent points of $U$ in $\mathcal{M}$ so if $R$ controls the distance between $\exp_{\iota(u)}(u_{\boldsymbol{m}})$ and $\iota(v_{\boldsymbol{m}})$ for all $\boldsymbol{m}$, then the assignment must be injective since the $R$-ball at $\iota(v_{\boldsymbol{m}})$ can contain at most one point of $\exp_p(U)$. Since $ M\ell^2\ll 1$ by assumption, $\ell^3\ll \ell/M$ and so $2\ell^3\ll \ell/M$ i.e. $\ell^3\ll \ell/M-\ell^3$ and we may take $R=\mathcal{O}(\ell^3)$. It is thus sufficient to prove that the recursion preserves the constraint:
	\begin{align}
			\rho_{\mathcal{M}}(\exp_{\iota(u)}(u_{\boldsymbol{m}}),\iota(v_{\boldsymbol{m}}))=\mathcal{O}(\ell^3)\nonumber.
	\end{align}
	Let us consider the case that $u_{\boldsymbol{m}}$ is in line with some sequence $u_1,...,u_n\in U$ and pick $v_{\boldsymbol{m}}$ in line with the corresponding sequence $v_1,...,v_n$ and $(M,\ell)$-adjacent to $v_n$. By construction, $\iota(v_n)$ is in line with $\exp_{\iota(u)}(\set{u_1,...,u_n})$ which ensures the desired constraint. Similar remarks go for joint neighbours ($(M,\ell)$-perpendicular points.) Note that we use the fact that perpendicular points in $S_{\ell,r}^D$ are separated by a distance $\mathcal{O}(\ell/M)$ while they are distorted by the exponential map a distance $\mathcal{O}(\ell^3)$ so the angular error is $\mathcal{O}(ell^3)/\mathcal{O}(\ell/M)=\mathcal{O}(M\ell^2)$. However since all points lie in the sphere we only have to propagate this angular error a distance $\mathcal{O}(\ell/M)$ leading to an overall error in the distance $\mathcal{O}(\ell^3)$.
	
	Note that the constraint $2M^{D-1}\leq \mathcal{O}(r\ell^{D-1}/\varepsilon^D)$ is required to ensure that it is indeed possible to pick a subset of $S_{\ell,r}^G(u)$ with $2M^{D-1}$ points.
\end{proof}
\begin{proof}[Proof of Corollary \ref{corollary: TraceApproximation}]
	This is a simple consequence of subadditivity and the definition of $\theta_u$. In particular we have
	\begin{align}
		\left|\text{ tr}_D(T)-\text{ tr}_G^u(f)\right|&\leq |\text{ tr}_D(T)-\text{ tr}_U(T)|\nonumber\\
		&\qquad +|\text{ tr}_U(T)-\text{ tr}_G^u(g)|\nonumber\\
		&=\mathcal{O}(M^{-1})+|\text{ tr}_U(T)-\text{ tr}_G^u(g)|\nonumber.
	\end{align}
	But, assuming that $\theta_u$ is well formed, $\theta_u$ is injective on the set $\theta_u^{-1}(U)$ and we have a well-defined inverse $v_{\textbf{m}}=\theta_u^{-1}(u_{\textbf{m}})$ for each multi-index $\textbf{m}$; thus we see that
	\begin{align}
		\sum_{v\in S_{\ell,r}^G(u)}\mathfrak{d}(\theta_u(v))g(v)&=\sum_{\textbf{m}}\mathfrak{d}(\theta_u(v_{\textbf{m}}))g(v_{\textbf{m}})\nonumber\\
		&=\sum_{\textbf{m}}\mathfrak{d}(u_{\textbf{m}})g(v_{\textbf{m}})\nonumber,
	\end{align}
	since by construction $\theta_u(v)=0$ (and hence $\mathfrak{d}(\theta_u(v))=0$) if $v\notin \theta_u^{-1}(U)$. Then by subadditivity we have:
	\begin{widetext}
	\begin{align}
		\left|\text{ tr}_D(T)-\text{ tr}_G^u(f)\right|&\leq\mathcal{O}(M^{-1})+\frac{\pi^{D-1}}{\omega_D M^{D-1}}\left|\sum_{\textbf{m}}\mathfrak{d}(u_{\textbf{m}})(T(u_{\textbf{m}},u_{\textbf{m}})-g(v_{\textbf{m}}))\right|\nonumber\\
		&\leq\mathcal{O}(M^{-1})+\frac{\pi^{D-1}}{\omega_D M^{D-1}}\sum_{\textbf{m}}\mathfrak{d}(u_{\textbf{m}})\left|T(u_{\textbf{m}},u_{\textbf{m}})-g(v_{\textbf{m}})\right|\nonumber\nonumber\\
		&=\mathcal{O}(M^{-1})+\mathcal{O}(\sigma)\nonumber
	\end{align}
	\end{widetext}
	which proves the statement.
\end{proof}
	\section{The Error from the Ollivier Curvature}\label{appendix: Semidiscrete}
	In this section we find the error associated with the Ollivier curvature. The basic result is the following lemma:
	\begin{lemma}\label{lemma: OllivierCurvatureMainTheorem}
	Let $G$ be a graph and let $\iota:G\hookrightarrow \mathcal{M}$ be an $\varepsilon$-isometry. For any $u\in G$, let $p\in \mathcal{M}$ be such that $p=\iota(u)$. Then given $\delta,\:\ell,\:r>0$ such that 
	\begin{align}
		\delta^4\ll \varepsilon\ll \min(\delta^2\ell^2,\delta^3) && \delta\ll \ell,
	\end{align}
	the mapping 
	\begin{align}
	    f:v\mapsto\frac{2(D+2)}{\delta^2} \kappa_G^\delta(u,v)
	\end{align}
	is a mapping on $S_{\ell,r}^G(u)$ such that
	\begin{align}
		|\text{Ric}(\exp_p^{-1}(\iota(v)),\exp_p^{-1}(\iota(v)))-f(v)|=\mathcal{O}\left(\sigma\right) 
	\end{align}
	where
	\begin{align}
	    \sigma=\max\left(\ell,\frac{\varepsilon}{\delta^2\ell^2},\frac{\varepsilon}{\delta^3}\right),
	\end{align}
	for all $v\in S_{\ell,r}^G(u)$.
\end{lemma}
	
	Recall that we have a graph $G$, a compact Riemannian manifold $\mathcal{M}$ and a $\varepsilon$-isometry $\iota:G\rightarrow \mathcal{M}$.
	
	We begin by defining
	\begin{align}
	    \delta \kappa(u,v)&= \frac{2(D+2)}{\delta^2}\left|\kappa_G^\delta(u,v)-\kappa_{\mathcal{M}}^\delta(\iota(u),\iota(v))\right|\nonumber\\
	    &=\frac{2(D+2)}{\delta^2\ell}|\mathcal{T}_G^\delta(m_u^\delta,m_v^\delta)-\mathcal{T}_{\mathcal{M}}^\delta(\mu_{\iota(u)}^\delta,\mu_{\iota(v)}^\delta)|\nonumber\\
	    &\qquad +\mathcal{O}\left(\frac{\varepsilon}{\delta^2\ell}\right)
	\end{align}
	where we have used the fact that ${\rm dis}(\iota)\leq \varepsilon$ and $r\ll \ell$ in moving to the second line. We remark that
	\begin{align}
	    |\mathcal{T}_G(m_u^\delta,m_v^\delta)-\mathcal{T}_{\mathcal{M}}(\mu_{\iota(u)}^\delta,\mu_{\iota(v)}^\delta)|=\mathcal{O}(\varepsilon)
	\end{align}
	for all pairs $(u,v)\in G\times G$ such that $u\in G$ and $v\in S^G_{\ell,r}(u)$ since $\iota$ is a $\mathcal{O}(\varepsilon)$-net so
	\begin{align}
	    \delta\kappa(u,v)&=\frac{2(D+2)}{\delta^2\ell}|\mathcal{T}_{\mathcal{M}}^\delta(\iota_*m_u^\delta,\iota_*m_v^\delta )-\mathcal{T}_{\mathcal{M}}^\delta(\mu_{\iota(u)}^\delta,\mu_{\iota(v)}^\delta)|\nonumber\\
	    &\qquad +\mathcal{O}\left(\frac{\varepsilon}{\delta^2\ell}\right).
	\end{align}
	$\delta\kappa$ is thus small as long as 
	\begin{align}
	    \varepsilon&\ll \delta^2\ell \nonumber |\mathcal{T}_{\mathcal{M}}^\delta(\iota_*m_u^\delta,\iota_*m_v^\delta )-\mathcal{T}_{\mathcal{M}}^\delta(\mu_{\iota(u)}^\delta,\mu_{\iota(v)}^\delta)|&\ll \delta^2\ell
	\end{align}
	for all relevant pairs $(u,v)$; by the triangle inequality the latter follows as long as
	\begin{align}\label{error: Semidiscrete}
	    \mathcal{T}_{\mathcal{M}}^\delta(\iota_*m_u^\delta,\mu_{\iota(u)}^\delta)\ll \delta^2\ell
	\end{align}
	for all $u\in G$. Thus showing convergence of the Ollivier curvature has reduced to showing that the \textit{semidiscrete discrepancy} $\mathcal{T}_{\mathcal{M}}^\delta(\iota_*m_u^\delta,\mu_{\iota(u)}^\delta)$ is small.
	
	For any $u\in G$, let $A_u$ be the set $\iota(B_\delta^G(u))_\varepsilon$, i.e. the $\varepsilon$-thickening of the set $\iota(B_\delta^G(u))$. $A_u$ is measurable set and let $\mu_{A}^u$ denote the uniform measure on $A_u$:
	\begin{align}
	    \mu_A^u(E)=\frac{{\rm vol}(E\cap A_u)}{{\rm vol}(A_u)}.
	\end{align}
	By subadditivity we have
	\begin{align}
	    \mathcal{T}_{\mathcal{M}}^\delta(\iota_*m_u^\delta,\mu_{\iota(u)}^\delta)\leq \mathcal{T}_{\mathcal{M}}^\delta(\iota_*m_u^\delta,\mu_A^u)+\mathcal{T}_{\mathcal{M}}^\delta(m_A^u,\mu_{\iota(u)}^\delta)
	\end{align}
	so the inequality \ref{error: Semidiscrete} holds as long as the analogous inequality holds for each Wasserstein distance on the right-hand side respectively. 
	
	Since $\iota(G)$ is a $\varepsilon$-net in $\mathcal{M}$ we have that
	\begin{align}
	    B_{\delta-\varepsilon}^{\mathcal{M}}(\iota(u))\subseteq A_u\subseteq B_{\delta+\varepsilon}^{\mathcal{M}}(\iota(u))
	\end{align}
	so
	\begin{align}
	    \mathcal{T}_{\mathcal{M}}^\delta(m_A^u,\mu_{\iota(u)}^\delta)\leq \mathcal{T}_{\mathcal{M}}^\delta(\mu_{\iota(u)}^{\delta+\varepsilon},\mu_{\iota(u)}^{\delta-\varepsilon}).
	\end{align}
	We may bound the right-hand side of the above, however, by constructing an explicit transport plan from $\mu_{\iota(u)}^{\delta+\varepsilon}$ to $\mu_{\iota(u)}^{\delta-\varepsilon}$. One obvious candidate is as follows:
	\begin{enumerate}
	    \item We assume we have dirt distributed according to $\mu_{\iota(u)}^{\delta+\varepsilon}$. We leave all of the dirt on the subset $B_{\delta-\varepsilon}^{\mathcal{M}}(\iota(u))$ where it is. Since the dirt is not moving this contributes nothing to the total transport cost.
	    \item Since the total amount of earth is normalised, the remaining dirt on $B_{\delta+\varepsilon}^{\mathcal{M}}(\iota(u))\backslash B_{\delta-\varepsilon}^{\mathcal{M}}(\iota(u))$ is to be spread evenly on $B_{\delta-\varepsilon}^{\mathcal{M}}(\iota(u))$. However this is done, the contribution to the cost can be kept to order
	    \begin{align}
	        \mathcal{O}\left((\delta+\varepsilon)\frac{{\rm vol}(B_{\delta+\varepsilon}^{\mathcal{M}}(\iota(u))\backslash B_{\delta-\varepsilon}^{\mathcal{M}}(\iota(u)))}{B_{\delta-\varepsilon}^{\mathcal{M}}(\iota(u))}\right)=\mathcal{O}(\varepsilon).
	    \end{align}
	\end{enumerate}
	Thus $\mathcal{T}_{\mathcal{M}}^\delta(\mu_{\iota(u)}^{\delta+\varepsilon},\mu_{\iota(u)}^{\delta-\varepsilon})=\mathcal{O}(\varepsilon)$ which is sufficiently small since we have assumed $\varepsilon\ll \delta^2\ell$.
	
	It thus remains to show that
	\begin{align}
	    \mathcal{T}_{\mathcal{M}}(\iota_*m_u^\delta,\mu_A^u)\ll \delta^2\ell. 
	\end{align}
	To do this, let us first define $A= \exp_u^{-1}({\rm supp}(\iota_*m_u^\delta))\subseteq \mathbb{R}^D$ and $B= \exp_{\iota(u)}^{-1}(A_u)\subseteq \mathbb{R}^D$ and let $m_A$ and $\mu_B$ denote the uniform Euclidean measures on $A$ and $B$ respectively, i.e.
	\begin{align}
	    m_A(E)=\frac{|A\cap E|}{|A|} && \mu_B(E)=\frac{\lambda(B\cap E) }{\lambda(B)}
	\end{align}
	where $\lambda$ is the $D$-dimensional Lebesgue measure. By construction $(\exp_{\iota(u)})_*m_A=\iota_*m_u^\delta$ and $(\exp_{\iota(u)})_*\mu_B=\mu_A^u$. Thus, by equation \ref{equation: DistortionExponential} and the fact that the pushforward of a transport plan between two measures is a transport plan between the pushforwards of those measures, we see that
	\begin{align}
	    \mathcal{T}_{\mathcal{M}}(\iota_*m_u^\delta,\mu_A^u)\leq \mathcal{T}_{D}(m_A,\mu_B)+\mathcal{O}(\delta^3).
	\end{align}
	This is sufficiently small as long as
	\begin{align}
	    \delta\ll \ell  && \mathcal{T}_{D}(m_A,\mu_B)\ll \delta^2\ell. 
	\end{align}
	Note that $\mathcal{T}_D= \mathcal{T}_{\mathbb{R}^D}$. Let $n= |A|$ and let $\tilde{n}$ be a positive integer such that $\lfloor\tilde{n}\omega_D(\delta-\varepsilon)^D(\delta+\varepsilon)^{-D}(1-\mathcal{O}(\varepsilon))\rfloor=n$. The idea is that if we evenly and deterministically distribute $\tilde{n}$ points in the cube $(-\delta-\varepsilon,\delta+\varepsilon)^D$, $n(1-\mathcal{O}(\varepsilon))$ of those points will lie in the ball $B_{\delta-\varepsilon}^D$, and the fraction $n/\tilde{n}$ is the same as the fraction $\lambda(B)/\lambda((-\delta-\varepsilon,\delta+\varepsilon)^D)$. Also $\tilde{n}=\mathcal{O}(n)$. The minimal distance between points in the even grid is $\mathcal{O}(n^{-1/D})$ trivially. Thus the minimal distance between points in the pushforwards of the grid of points under the exponential map is $\mathcal{O}(n^{-1/D}-\delta^3)$ so if $\varepsilon\ll n^{-1/D}-\delta^3$ there is at most one point of the grid within a distance $\varepsilon$ of some point of $\exp_u(A)$. Moreover, since $\iota(G)$ is a $\varepsilon$-net in $\mathcal{M}$ the grid points in $B_{\delta-\varepsilon}^D$ under the exponential map all lie within a distance $\varepsilon$ of some point of $B_{\delta-\varepsilon}^{\mathcal{M}}(\iota(u))\cap \iota(G)$. This defines a one-to-one correspondence between grid points in $B_{\delta-\varepsilon}^D$ and elements of $A\cap B_{\delta-\varepsilon}^D$, which in turn essentially defines a deterministic transport plan between the empirical measure on the grid points intersecting with $B$ and the empirical measure on $A$ (elements of $A\backslash B_{\delta-\varepsilon}^D$ are matched to random grid points). The cost of this plan will be $\mathcal{O}((1-\varepsilon)\delta^3+\varepsilon\delta)=\mathcal{O}(\delta^3)$. Thus we can consider the transport cost from the empirical measure on the intersection of the uniform grid with $A$ to $\mu_B$; but since $\varepsilon\ll n^{-1/D}$ this is less than the transport cost of the grid points to the uniform measure on $\mathcal{V}$ where $\mathcal{V}$ is the union over grid points of the Voronoi cells centred at each such grid-point. It is known that the transport plan that sends each Voronoi cell to its centre is optimal \cite{HartmannSchumacher_Semidiscrete} and has transport cost $\mathcal{O}(n^{-1/D})$. With this we finally find that we may approximate the Ollivier curvature on $\mathcal{M}$ by the Ollivier curvature on $G$ as long as
	\begin{align}
	    \varepsilon+\delta^3\ll n^{-\frac{1}{D}}\ll \delta^2 \ell  && \delta\ll \ell 
	\end{align}
	Noting that
	\begin{widetext}
	\begin{align}
	    \mathcal{O}\left(\frac{\delta^D}{\varepsilon^D}\right)=\left\lfloor\frac{{\rm vol}(B_{\delta}^{\mathcal{M}}(\iota(u)))}{{\rm vol}(B_{\varepsilon}^{\mathcal{M}}(\iota(u)))}\right\rfloor \leq n\leq N-\left\lfloor\frac{{\rm vol}(\mathcal{M}\backslash  B_{\delta}^{\mathcal{M}}(\iota(\varepsilon)))}{{\rm vol}(B_\varepsilon^{\mathcal{M}}(\iota(u)))}\right\rfloor =\mathcal{O}\left(\frac{\delta^D}{\varepsilon^D}\right)
	\end{align}
	\end{widetext}
	we see that $n^{-1/D}=\mathcal{O}(\varepsilon/\delta)$ and $n^{-1/D}\ll \delta^2\ell$ if $\varepsilon\ll \delta^3\ell\ll \delta^3$. Thus $\varepsilon+\delta^3\ll n^{-1/D}$ iff $\delta^3\ll n^{-1/D}$, i.e. $\delta^4\ll \varepsilon$. Thus convergence of the Ollivier curvature holds as long as
	\begin{align}
	    \delta^4\ll \varepsilon\ll \delta^3\ell
	\end{align}
	where $\delta\ll \ell$ necessarily as a result of the above.
	\section{Proof of Theorem \ref{theorem: MainTheorem}}\label{appendix: MainTheorem}
	We now prove the main theorem.
    \begin{proof}[Proof of Theorem \ref{theorem: MainTheorem}]
	Given the conditions in lemma \ref{lemma: OllivierCurvatureMainTheorem}, i.e.
	\begin{align}
		\delta^4\ll \varepsilon\ll \min(\delta^2\ell^2,\delta^3) && \delta\ll \ell,\nonumber
	\end{align}
	and $r\ll \ell$, we have
	\begin{align}
		\frac{2(D+2)}{\delta^2}\kappa_G^{\delta}(u,v)&=\nonumber\text{Ric}(\exp_p^{-1}(\iota(v)),\exp_p^{-1}(\iota(v)))+\delta\kappa \nonumber
	\end{align}
	where
	\begin{align}
		\delta\kappa=\nonumber\mathcal{O}\left(\max\left(\ell,\frac{\varepsilon}{\delta^2\ell^2},\frac{\varepsilon}{\delta^3}\right)\right).
	\end{align}
	But by theorem \ref{theorem: AngleAlgorithm} we have an angular assignment $\theta_u$ on $n=2M^{D-1}$ vertices at $u\in G$ such that
	\begin{align}
		\rho_{\mathcal{M}}(\exp_{\iota(u)}u_{\boldsymbol{m}},\iota(v_{\boldsymbol{m}}))=\mathcal{O}(\ell^3)\nonumber
	\end{align}
	as long as
	\begin{align}
		M\ell^2\ll 1 && 2M^{D-1}=\mathcal{O}(r\ell^{D-1}/\varepsilon^D)\nonumber.
	\end{align}
	Then 
	\begin{align}
		\rho_{D}(u_{\boldsymbol{m}},v_{\boldsymbol{m}})=\mathcal{O}(\ell^3)\nonumber
	\end{align} 
	and we may Taylor expand $\text{Ric}$ about $\theta_u(v)$ for each $v\in S_{\ell,r}^G(u)$ to obtain
	\begin{align}
		\text{Ric}(\exp_p^{-1}(\iota(v)),\exp_p^{-1}(\iota(v)))=\text{Ric}(\theta_u(v),\theta_u(v))+\mathcal{O}(\ell^3)\nonumber
	\end{align}
	i.e.
	\begin{align}
		\left|\kappa_G^{\delta}(u,v)-\text{Ric}(\theta_u(v),\theta_u(v))\right|=\mathcal{O}(\max(\ell^3,\delta \kappa))\nonumber.
	\end{align}
	But noting that
	\begin{align}
		\mathcal{A}_{DEH}(G;\delta,\ell,r)=\frac{\omega_D}{N}\sum_{u\in G}\text{tr}_G^u(f)\nonumber,
	\end{align}
	we have by corollary \ref{corollary: TraceApproximation} that
	\begin{align}
		\mathcal{A}_{DEH}(G;\delta,\ell,r)&=\frac{\omega_D}{N}\sum_{u\in G}\text{tr}_G^u(f)\nonumber\\
		&=\frac{\omega_D}{N}\sum_{u\in G}\text{tr}_D(\text{Ric}_{\iota(u)})\nonumber\\
		&\qquad +\mathcal{O}(\max(\ell^3,\delta \kappa,M^{-1}))\nonumber\\
		&=\frac{\omega_D}{N}\sum_{u\in G}R(\iota(u))\nonumber\\
		&\qquad +\mathcal{O}(\max(\ell^3,\delta \kappa,M^{-1}))\nonumber
	\end{align}
	where $R$ is the scalar curvature. The statement then follows by lemma \ref{lemma: ApproximatingIntegrals} and the definition of $\mathcal{A}_{EH}=\text{vol}_{\mathcal{M}}(R)$ as long as $\varepsilon=N^{-1/D}$.
	
	We have thus derived the desired result given the following constraints:
	\begin{align}
			\varepsilon&=N^{-\frac{1}{D}} & \delta^4&\ll \varepsilon\ll \min(\delta^2\ell^2,\delta^3) \nonumber\\ \varepsilon&\ll \delta,\: r\ll \ell & M\ell^2&\ll 1\ll M \nonumber\\ 2M^{D-1}&=\mathcal{O}\left(\frac{r\ell^{D-1}}{\varepsilon^D}\right)\nonumber.
	\end{align}
	Note that these constraints automatically ensure that the error vanishes as $N\rightarrow \infty$. We show that given the definitions \ref{equation: Scales}, the constraints above follow from the constraints \ref{constraint: Scales}: $\varepsilon=N^{-1/D}$ is directly stated in one of these constraints as required. In particular it is easy to see that $\delta^4\ll \varepsilon$ comes from $1<4aD$ which is assumed in constraint \ref{constraint: Scales}. On the other hand $\varepsilon\ll \delta^3$ is equivalent to $3aD<1$. The condition $\varepsilon\ll \delta^2\ell^2$ is equivalent to $2(a+b)<1$. Also $\delta\ll \ell$ implies $a>b$ so both of these constraints follow from $(3a+b)<1$, which is also assumed in \ref{constraint: Scales}. $\varepsilon\ll r\ll \ell$ gives $c\ll b<1/D$ while $a<1/D$ implies $\varepsilon\ll \delta$. It thus remains to derive the constraints for $M$. But if $0<d$, $1\ll M$ while $M\ell^2\ll 1$ if $d<2b$. Also $M=\mathcal{O}(r\ell^{D-1}/\varepsilon^D)$ is equivalent to 
	\begin{align}
		d(D-1)\leq 1-(D-1)b-c\text{ or } b\leq \frac{1-c}{D-1}-d\nonumber
	\end{align}
	where we have used the fact that $d>0$. This gives us the final of the constraints in equation \ref{constraint: Scales}.
	
	Finally it remains to show that the constraints \ref{constraint: Scales} are consistent. Pick some $\epsilon>0$ such that $\epsilon<1/4$ and let
	\begin{align}
		a=c=\frac{1+\epsilon}{4D} && b=\frac{1-4\epsilon}{4D}.\nonumber
	\end{align}
	Clearly, $0<b<a=c<1/D$ trivially. For the second constraint note that.
	\begin{align}
		(3a+b)D=1-\frac{1}{4}\epsilon<1<1+\epsilon=4aD\nonumber
	\end{align}
	as required. For the final inequality we first note that since $b>0$ we can pick any $d>0$ such that $d<(1-4\varepsilon)/2D$ to obtain the desired result here. On the other hand, since $d>0$, $(1-c)/(D-1)-d>(1-c)/(D-1)-2b$ and the desired inequality certainly holds if $3b(D-1)<(1-c)$. Multiplying both sides by $4D$ we obtain
	\begin{align}
		3(D-1)(1-4\epsilon)<4D-1-\epsilon\text{ or }(13-12D)\epsilon<D+2. \nonumber
	\end{align}
	The left-hand side of the second expression is negative for all $D>1$ while the right-hand side is positive so the desired inequality holds for arbitrary choices of $\epsilon>0$. On the other hand in $D=1$ we required $\epsilon<3$ which has already been imposed by the choicce $\epsilon<1/4$.
    \end{proof}
    \bibliography{Ref}
\end{document}